\documentclass[reprint,pra, twocolumn, amsmath,amssymb,
 aps,nofootinbib,superscriptaddress,floatfix]{revtex4-2}

\usepackage{graphicx}
\usepackage{dcolumn}
\usepackage{bm}
\usepackage[dvipsnames]{xcolor}
\usepackage{amsthm}
\usepackage{hyperref}

\newtheorem{prop}{Proposition}
\newtheorem{lem}{Lemma}

\newcommand{\I}{\mathbb{I}}
\newcommand{\bC}{\mathbb{C}}
\newcommand{\cH}{\mathcal{H}}
\newcommand{\tA}{\tilde{A}}
\newcommand{\tB}{\tilde{B}}
\newcommand{\tpsi}{\tilde{\psi}}
\newcommand{\junk}{\text{junk}}

\def \diracspacing {0.7pt}
\newcommand{\bra}[1]{\langle #1 \hspace{\diracspacing} |} 
\newcommand{\ket}[1]{| \hspace{\diracspacing} #1 \rangle} 
\newcommand{\braket}[2]{\langle #1 \hspace{\diracspacing} | \hspace{\diracspacing} #2 \rangle} 
\newcommand{\ketbra}[2]{| \hspace{\diracspacing} #1 \rangle \langle #2 \hspace{\diracspacing} |} 
\newcommand{\ketbraq}[1]{\ketbra{#1}{#1}} 

\DeclareMathOperator{\tr}{tr}
\newcommand{\norm}[2][]{#1\left| \! #1\left| #2 #1\right| \! #1\right|}
\newcommand{\abs}[2][]{#1| #2 #1|}
\begin{document}

\preprint{APS/123-QED}

\title{Extending quantum correlations to arbitrary distances via parallel repetition of routed Bell tests}

\author{Anubhav Chaturvedi}
\email{anubhav.chaturvedi@pg.edu.pl}
\affiliation{Faculty of Applied Physics and Mathematics,
             Gdańsk University of Technology,
             Gabriela Narutowicza 11/12, 80‑233 Gdańsk, Poland}
\affiliation{International Centre for Theory of Quantum Technologies (ICTQT),
             University of Gdańsk, 80‑308 Gdańsk, Poland}

\author{Marcin Pawłowski}
\affiliation{International Centre for Theory of Quantum Technologies (ICTQT),
             University of Gdańsk, 80‑308 Gdańsk, Poland}

\noaffiliation            

\author{Máté Farkas}
\email{mate.farkas@york.ac.uk}
\affiliation{Department of Mathematics, University of York,
             Heslington, York YO10 5DD, United Kingdom}

\date{\today}
\begin{abstract}
Applications such as Device-Independent Quantum Key Distribution (DIQKD) require loophole-free certification of long-distance quantum correlations. However, these distances remain severely constrained by detector inefficiencies and unavoidable transmission losses. To overcome this challenge, we consider parallel repetitions of the recently proposed routed Bell experiments, where transmissions from the source are actively directed either to a nearby or a distant measurement device. We analytically show that the threshold detection efficiency of the distant device--needed to certify non-jointly-measurable measurements, a prerequisite of secure DIQKD--decreases exponentially, optimally, and robustly, following $\eta^*=1/2^N$, with the number $N$ of parallel repetitions.
\end{abstract}

\maketitle
\section{Introduction}
Quantum theory enables spatially separated parties to share nonlocal correlations, powering several classically impossible applications \cite{Bell1964,Brunner2014review}. One of the most notable examples of such applications is device-independent quantum key distribution (DIQKD), which enables distant parties to establish a secure cryptographic key without the need to trust the devices used \cite{Ekert1991, mayers1998quantum, BHK05, Acin2007, PABGMS09}. To employ DIQKD in practice, nonlocal correlations over long distances are needed. 
The most persistent challenge in the experimental realization of long-range nonlocal correlations stems from limited efficiency of particle detectors and the inevitable losses during the transmission of quantum systems \cite{Pearle1970,PhysRevD.10.526}. Certifying nonlocal quantum correlations requires the overall efficiency $\eta$ (accounting for all losses and detector imperfections) to exceed a certain threshold $\eta^*$. 
This threshold is typically very high, severely limiting the distances over which nonlocal correlations
can be reliably certified. 

Several approaches have been proposed to lower the critical detection efficiency requirement \cite{Eberhard1993,Massar2002,Massar2003,PhysRevLett.104.060401,Vertesi2010,XuLatest}. In this work, we focus on two of the most promising recently proposed methods—parallel repetition \cite{miklin2022exponentially} and routed Bell experiments \cite{chaturvedi2024extending}. In parallel repetition, a single pair of particles encodes $N$ entangled states \cite{Kwiat01111997}, allowing the simultaneous violation of $N$ copies of a Bell inequality in parallel. This approach can drive $\eta^*$ down in proportion to $\big(\frac{C}{Q}\big)^N$, where $C$ is the classical bound and $Q>C$ is the optimal quantum value of the Bell inequality in question. In the limit of infinite repetitions this leads to a ``Bell violation in a single shot'', since $\eta^*\to0$ as $N\to\infty$ \cite{miklin2022exponentially,araujo2020bell}. However, the actual decay rate is often suboptimal
and significant gains generally require large $N$ \cite{miklin2022exponentially}.

In routed Bell experiments \cite{chaturvedi2024extending}, depicted in FIG. \ref{routedSetup},  Alice's device $A$ remains near the source while Bob switches between a nearby device $B_0$ and a distant one $B_1$, the latter of which has an overall efficiency $\eta<1$. This setup resembles the DI protocols with a ``local Bell test" \cite{PhysRevX.3.031006, PhysRevLett.133.120803}. In routed Bell experiments, strong nonlocal correlations established close to the source (between $A$ and  $B_0$) help characterize Alice's measurement device and the source, thereby lowering the critical efficiency $\eta^*$ of the distant device $B_1$. 

If the devices near the source are perfect, the impossibility of describing $B_1$ in terms of local hidden variables-can be certified at arbitrary low detection efficiency \cite{chaturvedi2024extending}. However, this is not enough to ensure security of DIQKD between $A$ and $B_1$ against an active eavesdropper, as  for secure DIQKD, we need to ensure that the measurements in $B_1$ are non-jointly measurable (NJM) \cite{chaturvedi2024extending,Lobo2024certifyinglongrange,PhysRevLett.133.120803}. 
An eavesdropper could intercept the transmission to $B_1$, perform a joint measurement, then forward the results to $B_1$, making $B_1$'s apparent randomness known to the eavesdropper (see FIG. \ref{fig:enter-label2}). Therefore, to guarantee the security of DIQKD based on routed Bell experiments, the correlations between $A$ and $B_1$ must be incompatible with any jointly measurable (JM) model at $B_1$. In other words, the measurements of $B_1$ must be certifiably NJM \cite{Lobo2024certifyinglongrange,Masini2024jointmeasurability}.

Correlations explainable and incompatible with a JM model at $B_1$ are termed short-range quantum (SRQ) and long-range quantum (LRQ)  correlations, respectively, in \cite{Lobo2024certifyinglongrange}. Moreover, SRQ correlations are equivalently those that are compatible with an entanglement-breaking channel between the switch and $B_1$. In the case of binary measurements, Lobo, Pauwels and Pironio showed that the critical detection efficiency for certifying NJM measurements in $B_1$ reduces to $\eta^*=\frac{1}{2}$ \cite{Lobo2024certifyinglongrange} in a routed scenario. 
 Unlike the standard CHSH Bell scenario, where approaching (but not reaching) $\eta^* = \frac{1}{2}$ demands a fragile, near-product two-qubit state \cite{Eberhard1993,Eberhard1995,gigena2024robustselftestingbellinequalities}, routed Bell experiments achieve $\eta^*=\frac{1}{2}$ using the maximally entangled two-qubit state \cite{Lobo2024certifyinglongrange}. Remarkably, this leads to $\eta^*=\frac12$ for secure DIQKD based on routed Bell experiments \cite{PhysRevLett.133.120803,roydeloison2024deviceindependent}.

Two distinct strategies achieve $\eta^*=1/2$, and we consider their $N$-fold parallel repetitions. Building on self-testing statements for the maximal quantum violation of $N$-product CHSH inequalities \cite{Col17,McKague17,Supic2021deviceindependent} between $A$ and $B_0$, we show analytically that the critical detection efficiency for certifying NJM measurements in $B_1$ drops exponentially with $N$, specifically $\eta^*=1/2^N$. This scaling is tight, since $B_1$ has $2^N$ settings, and therefore the critical detection efficiency for certification of NJM measurements in $B_1$ and hence for secure DIQKD cannot be lower than $\eta^*\geq 1/2^N$ \cite{chaturvedi2024extending,LPP24,roydeloison2024deviceindependent}. We further demonstrate that the robustness of the self-testing statements 
translates into robustness for our results. Finally, we use non-commutative polynomial optimization (NPO) \cite{Navascues2007,Navas_2008,PNA10}, to demonstrate that $2$-fold parallel repeated strategies exhibit increased resilience to imperfections in the source as compared to their single copy counterparts.

\begin{figure}
    \centering
    \includegraphics[width=\linewidth]{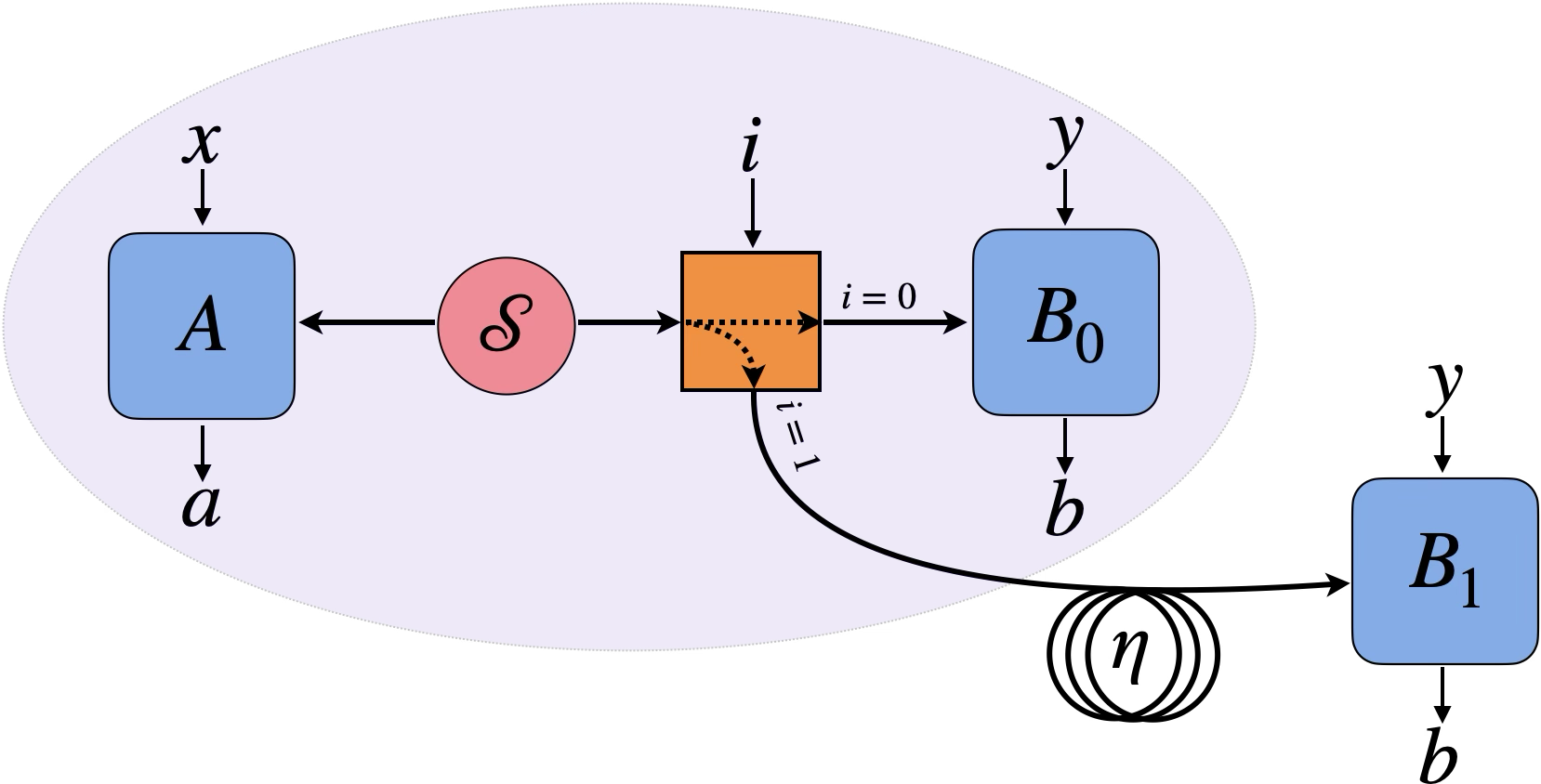}
    \caption{\emph{Routed Bell experiments}: A source distributes entangled particles to Alice ($A$) and Bob. A switch redirects Bob's subsystem depending on the input bit $i$  to either Bob's measurement device $B_0$ located near the source or to $B_1$ placed away from the source, such that the effective detection efficiency at $B_1$ is $\eta<1$. }
    \label{routedSetup}
\end{figure}
\section{Prelminaries}
In this section, we revisit the key concepts relevant to our discussion: routed Bell experiments, non-jointly measurable (NJM) measurements, self-testing statements, and the $r\text{BB84}$ and $r\text{CHSH}$ strategies for routed Bell experiments. 

In each round of a routed Bell experiment \cite{chaturvedi2024extending} (see FIG. \ref{routedSetup}), the source distributes the state $\rho_{AB}$ to Alice and Bob. Based on her input $x$, Alice's performs a measurement described by a POVM $\{A^x_a\}_a$. Based on $i$, Bob's subsystem is routed to the measurement device $B_i$. Based on his input $y$, Bob's device $B_i$ performs a measurement described by a POVM $\{B^{y,i}_b\}_b$. The three tuple $(\rho_{AB},\{A^x_a\}_{a,x},\{B^{y,i}_b\}_{b,y,i})$ describes a quantum strategy for the routed Bell experiment. As a result the parties observe correlations described by the conditional probability distributions
\begin{equation}    p(a,b|x,y,i)= \tr(\rho_{AB}A^x_a\otimes B^{y,i}_b).
\end{equation}

To ensure security of DIQKD schemes based on routed Bell experiments, correlations between $A$ and $B_1$ must not be explainable via jointly measurable (JM) measurements in $B_1$ \cite{Masini2024jointmeasurability}. Specifically, the measurements in $B_1$ are JM if there exists a joint measurement described by the POVM $\{E_{\vec{b}}\}_{\vec{b}}$ with outcomes $\vec{b}\equiv (b_y)_y$, specifying the outcomes of $B_1$ in each round, that is
\begin{equation} \label{JM}
        p(a,b|x,y,1)\underset{\text{JM}({B_1})}{=} \sum_{\vec{b} }\delta_{b_y,b} \tr(\rho_{AB}A^x_a\otimes E_{\vec{b}}).
\end{equation}
If correlations between $A$ and $B_1$ can explained via JM measurements in $B_1$, as in the above equation, then an eavesdropper could intercept the transmission from the switch to $B_1$,
perform the parent measurement $\{E_{\vec{b}}\}_{\vec{b}}$, and forward the outcome $\vec{b}$ to dictate the outcomes of $B_1$, outputting $b = b_y$ whenever the input is $y$ (see FIG. \ref{fig:enter-label2}).
Correlations not captured by \eqref{JM}, certify that the measurements in $B_1$ are NJM, a pre-requisite of secure DIQKD.
\begin{figure}
    \centering
    \includegraphics[width=\linewidth]{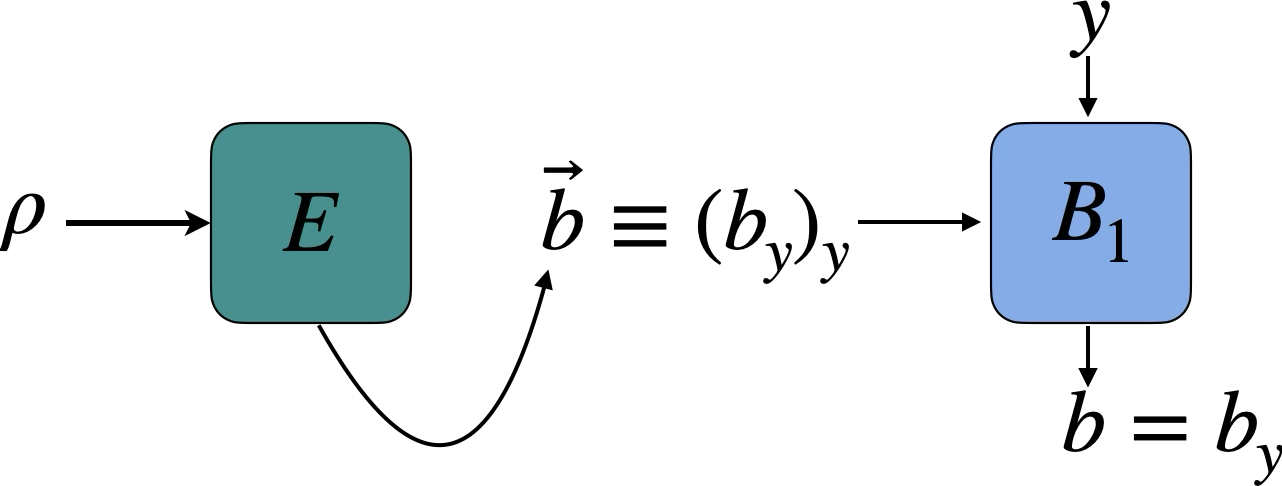}
    \caption{\emph{Active attack exploiting joint-measurability}: If the correlations between $A,B_1$ can be explained with jointly-measurable (JM) measurements in $B_1$ \eqref{JM}, then an eavesdropper $E$ can intercept the quantum signal $\rho$ intended for $B_1$, perform the joint (parent) measurement described by the POVM $\{E_{\vec{b}}\}_{\vec{b}}$, and transmit the classical outcome $\vec{b}=(b_y)_{y}$ to $B_1$. Upon receiving $\vec{b}$, $B_1$ then classically post-processes it to produce $b=b_y$ based on Bob's input $y$. Therefore, to ensure security of DIQKD against such an eavesdropper $B_1$ measurements must be certifiably non-jointly measurable (NJM).}
    \label{fig:enter-label2}
\end{figure}

In the simplest routed Bell experiment, with $x,i,y,a\in\{0,1\}$ and $b\in\{0,1,\emptyset\}$, where $\emptyset$ denotes the no-click event, the critical detection efficiency for the certification of NJM measurements in $B_1$ is $\eta^*=\frac12$, using a maximally entangled two-qubit state \cite{Lobo2024certifyinglongrange}. In particular, two distinct strategies, termed $r\text{BB84}$ and $r\text{CHSH}$, achieve this feat. In both cases, the devices close to the source, $A$ and $B_0$, witness the maximal violation of the CHSH inequality \cite{Cirelson1980}, that is,
\begin{align}\label{TsirelsonsB} \nonumber
    \mathcal{C}_{AB_0}&=\frac{1}{4}\sum_{x,y,a,b=0}^1\delta_{x\cdot y
    = a\oplus b}p(a,b|x,y,0)\\
    &=\frac{1}{2}\left(1+\frac{1}{\sqrt{2}}\right)=\alpha.
\end{align}
The maximal violation of the CHSH inequality \eqref{TsirelsonsB} self-tests the optimal quantum strategy $(\ket{\phi^+}_{AB},\{A^x_a\},\{B^y_b\})$ \cite{McKague_2012}, which means that, for any quantum realization $(\tilde{\rho}_{\tA \tB},\{\tilde{A}^x_a\},\{\tilde{B}^{y,0}_b\})$ on Hilbert spaces $\cH_{\tA}$ and $\cH_{\tB}$ achieving \eqref{TsirelsonsB}, there exist Hilbert spaces $\cH_{A'}$ and $\cH_{B'}$ and isometries $V_{\tilde{A}} : \cH_{\tA} \to \mathbb{C}^2 \otimes \cH_{A'}$ and $V_{\tilde{B}} : \cH_{\tB} \to \mathbb{C}^2 \otimes \cH_{B'}$ such that for every purification $\ket{\tilde{ \psi } }_{\tA \tB P} \in \cH_{\tA} \otimes \cH_{\tB} \otimes \cH_P$ of $\tilde{\rho}_{\tA \tB}$ we have
\begin{align} \label{CHSHselfTest}
(V_{\tilde{A}} \otimes V_{\tilde{B}} \otimes \I_P) (\tA^x_a \otimes \tB^y_b \otimes \I_P ) \ket{\tilde{\psi}_{\tA \tB P}} \\ \nonumber
= (A^x_a \otimes B^y_b) \ket{\phi^+} \otimes \ket{\junk}_{A'B'P}
\end{align}
for some state $\ket{\junk}_{A'B'P}$,
where $\ket{\phi^+}=\frac{1}{\sqrt{2}}(\ket{00}+\ket{11})$ and
\begin{align} \label{optMeasCHSH}
    A^x_a&=\frac{1}{2}\left[ \mathbb{I}_2+(-1)^a(\delta_{x,0}Z+\delta_{x,1}X)\right] \\ \nonumber
B^{y}_b&=\frac{1}{2}\left[ \mathbb{I}_2+\frac{(-1)^b}{\sqrt{2}}\bigg(\delta_{y,0}(X+Z)+\delta_{y,1}(X-Z)\bigg)\right],
\end{align}
where $\mathbb{I}_2,Z,X$ are the identity operator on $\mathbb{C}^2$ and the Pauli Z and X operators, respectively. 

The two strategies differ in the measurements for the distant device $B_1$. In the $r\text{BB84}$ strategy, $B_1$ copies Alice's device $A$ and measures the Pauli $Z$ and $X$ observables, that is, $B^{y,1}_b=A^{y}_{b}$ from Eq.~\eqref{optMeasCHSH}. When $\eta=1$, this strategy achieves maximum value of the BB84 inequality, that is
\begin{equation}\label{BB84game}
    \mathcal{B}_{AB_1}=\frac{1}{2}\sum_{x,y,a,b}\delta_{x,y}\delta_{a,b}p(a,b|x,y,1)=1.
\end{equation}
In the $r\text{CHSH}$ strategy, $B_1$ mimics Bob's device $B_0$ and measures the $X+Z$ and $X-Z$ observables, that is, $B^{y,1}_b=B^y_b$ from Eq.~\eqref{optMeasCHSH}. 
When $B_1$ operates perfectly, this strategy maximally violates the CHSH inequality between $A$ and $B_1$, that is,  $\text{CHSH}_{AB_1}=\alpha$. When the devices close to the source are perfect and $B_1$ clicks with efficiency $\eta<1$, keeping the no-click $\emptyset$ event as an additional outcome, the critical detection efficiency of both $r\text{BB84}$ and $r\text{CHSH}$ for the certification of NJM measurements in $B_1$ turns out to be $\eta^*=\frac{1}{2}$, which is tight since $\eta^*\geq 1/2$ for a device with two settings. Remarkably, for these strategies the threshold detection efficiency for a positive key rate in DIQKD also turns out to be $\eta^*\approx 1/2$ \cite{roydeloison2024deviceindependent}.

In the next section, we present the $N$-fold parallel repeated versions of the $r\text{BB84}$ and $r\text{CHSH}$ strategies. We then demonstrate that these strategies lead to exponentially and optimally decreasing critical detection efficiency, $\eta^*=\frac{1}{2^N}$, for the certification of NJM measurements in $B_1$.

\section{Parallel repeated routed Bell experiments}
In this section, we present the $N$-fold parallel repeated versions of the $r\text{BB84}$ and $r\text{CHSH}$  strategies for routed Bell experiments, termed $r\text{BB84}^{N}$ and $r\text{CHSH}^{N}$, respectively. 

Consider a Bell experiment where Alice and Bob both have $N$-bit strings $\bm{x}\equiv (x_j)^N_{j=1},\bm{y}\equiv (y_j)^N_{j=1}\in\{0,1\}^N$ as inputs and produce $N$-bit strings $\bm{a}\equiv (a_j)^N_{j=1},\bm{b}\equiv (b_j)^N_{j=1}\in\{0,1\}^N$ as outputs. In the routed version, based on an additional input bit $i\in\{0,1\}$, Bob decides whether to direct the transmission from the source to his measurement device $B_0$ or to $B_1$. 

In both the strategies, the source prepares a $2^{N}\times2^{N}$ dimensional state $\ket{\phi^+}^{\otimes N}_{AB}$, and the devices close to the source,  $A$ and $B_0$, perform the sharp measurements $\{A^{\bm{x}}_{\bm{a}}\},\{B^{ \bm{y},0 }_{ \bm{b} } = B^{ \bm{y} }_{ \bm{b} }\}$, respectively, where
\begin{align}\label{CHSHNmeas}
A^{\bm{x}}_{\bm{a}}=\bigotimes^N_{j=1}A^{x_j}_{y_j}, \ \ \ \ B^{\bm{y}}_{\bm{b}}=\bigotimes^N_{j=1} B^{y_j}_{b_j},
\end{align}
where $\{A^{x_j}_{a_j}\},\{B^{y_j}_{b_j}\}$ are two-dimensional projections defined in \eqref{optMeasCHSH}. This strategy attains the maximal violation of the $N$-product CHSH inequality, that is,
\begin{align}\label{CHSHNTB}
\mathcal{C}^{ N}_{AB_0}&=\frac{1}{2^{2N}}\sum_{\bm{x},\bm{y},\bm{a},\bm{b}}\delta_{\bm{x}\cdot\bm{y}=\bm{a}\oplus\bm{b}}p(\bm{a},\bm{b}|\bm{x},\bm{y},0)=\alpha^N,
\end{align}
where $\delta_{\bm{x}\cdot\bm{y}=\bm{a}\oplus\bm{b}}=\Pi^N_{j=1}\delta_{x_j\cdot y_j=a_j\oplus b_j }$.
Just like the CHSH inequality, the maximal violation of the $N$-product CHSH inequality \eqref{CHSHNTB} self-tests the optimal quantum strategy $(\ket{\phi^+}^{\otimes N}_{AB},\{A^{\bm{x}}_{\bm{a}}\},\{B^{\bm{y}}_{\bm{b}}\})$ \cite{Col17,McKague17,Supic2021deviceindependent}. In particular, \cite{Supic2021deviceindependent} implies that for any quantum realization $(\tilde{\rho}_{\tA \tB},\{\tilde{A}^{\bm{x}}_{\bm{a}}\},\{\tilde{B}^{\bm{y},0}_{\bm{b}}\})$
on Hilbert spaces $\cH_{\tA}$ and $\cH_{\tB}$ achieving \eqref{CHSHNTB}, there exist Hilbert spaces $\cH_{A'}$ and $\cH_{B'}$ and isometries $V_{\tilde{A}} : \cH_{\tA} \to \mathbb{C}^{2^N} \otimes \cH_{A'}$ and $V_{\tilde{B}} : \cH_{\tB} \to \mathbb{C}^{2^N} \otimes \cH_{B'}$ such that for every purification $\ket{\tilde{ \psi } }_{\tA \tB P} \in \cH_{\tA} \otimes \cH_{\tB} \otimes \cH_P$ of $\tilde{\rho}_{\tA \tB}$ we have
\begin{align} \label{CHSHNselfTest}
(V_{\tilde{A}} \otimes V_{\tilde{B}} \otimes \I_P) (\tA^{\bm{x}}_{\bm{a}} \otimes \tB^{\bm{y},0}_{\bm{b}} \otimes \I_P ) \ket{\tilde{\psi}_{\tA \tB P}} \\ \nonumber
= (A^{\bm{x}}_{\bm{a}} \otimes B^{\bm{y}}_{\bm{b}}) \ket{\phi^+}^{\otimes N} \otimes \ket{\junk}_{A'B'P}
\end{align}
for some state $\ket{\junk}_{A'B'P}$, where the $2^N$-dimensional projections $A^{\bm{x}}_{\bm{a}}$ and $B^{\bm{y}}_{\bm{b}}$ are defined in \eqref{CHSHNmeas}. In the following, for simplicity we use the notation $\cH_A = \cH_B = \mathbb{C}^{2^N}$. Note that by summing over $\bm{b}$ and $\bm{a}$, Eq.~\eqref{CHSHNselfTest} implies the relations
\begin{equation} \label{eq:CHSHN_selftest_sumab}
\begin{split}
(V_{\tilde{A}} & \left. \otimes V_{\tilde{B}} \otimes \I_P) (\tA^{\bm{x}}_{\bm{a}} \otimes \I_{\tB} \otimes \I_P ) \ket{\tilde{\psi}_{\tA \tB P}} \right. \\
& \left. = (A^{\bm{x}}_{\bm{a}} \otimes \I_B ) \ket{\phi^+}^{\otimes N} \otimes \ket{\junk}_{A'B'P} \right. \\
(V_{\tilde{A}} & \left. \otimes V_{\tilde{B}} \otimes \I_P) \ket{\tilde{\psi}_{\tA \tB P}} = \ket{\phi^+}^{\otimes N} \otimes \ket{\junk}_{A'B'P}
\right.
\end{split}
\end{equation}

Yet again, the two strategies, $r\text{BB84}^N$ and $r\text{CHSH}^N$, differ in the measurements for the distant device $B_1$. In the $r\text{BB84}^N$ strategy, $B_1$ copies Alice's device $M_1$ and measures products of Pauli $Z$ and $X$ observables, that is,
\begin{equation}
    B^{\bm{y},1}_{b}= A^{\bm{y}}_{\bm{b}},
\end{equation}
where the $A^{\bm{y}}_{\bm{b}}$ are $2^N$-dimensional projectors from \eqref{CHSHNmeas}. When $\eta=1$, the $r\text{BB84}^N$ strategy achieves the maximum value of the $N$-product BB84 inequality, 
\begin{equation}\label{BB84Ngame}
    \mathcal{B}^{ N}_{AB_1}=\frac{1}{2^N}\sum_{\bm{x},\bm{y},\bm{a},\bm{b}}\delta_{\bm{x}=\bm{y}}\delta_{\bm{a}=\bm{b}}p(\bm{a},\bm{b}|\bm{x},\bm{y},1)=1,
\end{equation}
where $\delta_{\bm{x}=\bm{y}}\delta_{\bm{a}=\bm{b}}=\Pi^N_{j=1}\delta_{x_j=y_j}\delta_{a_j=b_j}$.
In the $r\text{CHSH}^N$ strategy, $B_1$ performs the same measurements as $B_0$, that is, $B^{\bm{y},1}_b=B^{\bm{y}}_{\bm{b}}$, where $B^{\bm{y}}_{\bm{b}}$ are $2^N$-dimensional projectors from \eqref{CHSHNmeas}. When $B_1$ is perfect, the $r\text{CHSH}^N$ strategy attains the maximal quantum violation of the $N$-product CHSH inequality, that is, $\mathcal{C}^N_{AB_1}=\alpha^N$.

Since in both strategies the device $B_1$ has $2^N$ measurement settings, $\eta^*\geq 1/2^N$ for the certification of NJM measurements and hence, for secure DIQKD. In the next section, we demonstrate that this bound can be achieved with 
both $r\text{BB84}^N$ and $r\text{CHSH}^N$ strategies, that is, when the devices $A$ and $B_0$ are perfect, NJM measurements in $B_1$ can be certified whenever $\eta>\eta^*=1/2^N$.  

\section{Exponentially and tightly decreasing critical detection efficiencies}

In this section, we derive our main results, namely that NJM measurements in $B_1$ can be certified with the strategies $r\text{BB84}^N$ and $r\text{CHSH}^N$ whenever $\eta>1/2^N$, given the devices $A$ and $B_0$ are perfect. 

In both cases, the perfect devices $A$ and $B_0$ witness the maximal quantum violation of the $N$-product CHSH inequality \eqref{CHSHNTB}, which self-tests the reference quantum strategy \eqref{CHSHNselfTest}. From the self-testing statements we only need to carry forward two
implications. First, the self-testing statements \eqref{CHSHNselfTest} imply that Alice's marginals must be uniformly random, that is
\begin{equation} \label{uniformMarg}
    p_A(\bm{a}|\bm{x})=\frac{1}{2^N} \quad \forall \bm{a}, \bm{x}.
\end{equation}
Additionally, Eq.~\eqref{eq:CHSHN_selftest_sumab} lets us infer Bob's states $\{\tilde{\rho}^{(\tilde{B})}_{\bm{a}|\bm{x}}\}$, remotely prepared by Alice measuring $\bm{x}$ and observing the outcome $\bm{a}$, up to a local isometry $V_{\tilde{B}}$
\begin{align} \label{ImpImplication} \nonumber
& V_{\tilde{B}} \tilde{\rho}^{(\tilde{B})}_{\bm{a}|\bm{x}} V_{\tilde{B}}^\dagger = \frac{V_{\tilde{B}} \tr_{\tA P} \left[ \ketbraq{\tpsi} ( \tilde{A}^{\bm{x}}_{\bm{a}} \otimes \mathbb{I}_{\tilde{B}} \otimes \I_P ) \right] V_{\tilde{B}}^\dagger}{p_A(\bm{a}|\bm{x})}, \\ \nonumber
&=2^N \tr_{\tA P} \left[(\mathbb{I}_A\otimes V_{\tilde{B}} \otimes \I_P ) \ketbraq{\tpsi} (\tilde{A}^{\bm{x}}_{\bm{a}}\otimes V_{\tilde{B}}^\dagger \otimes \I_P )\right] \\ \nonumber
&=2^N \tr_{\tA P} \left[( V_{\tilde{A}}^\dagger V_{\tilde{A}} \otimes V_{\tilde{B}} \otimes \I_P ) \ketbraq{\tpsi} (\tilde{A}^{\bm{x}}_{\bm{a}}\otimes V_{\tilde{B}}^\dagger \otimes \I_P )\right] \\ \nonumber
&=2^N \tr_{A A' P} \left[( V_{\tilde{A}} \otimes V_{\tilde{B}} \otimes \I_P ) \ketbraq{\tpsi} (\tilde{A}^{\bm{x}}_{\bm{a}} V^\dagger_A \otimes V_{\tilde{B}}^\dagger \otimes \I_P )\right] \\ \nonumber
&=2^N \tr_{A A' P} \left[ \ketbraq{\phi^+}^{\otimes N} (A^{\bm{x}}_{\bm{a}} \otimes \I_B) \otimes \ketbraq{ \junk }_{A'B'P } \right] \\ \nonumber
&=2^N \tr_{A} \left[ \ketbraq{\phi^+}^{\otimes N} (\I_A \otimes A^{\bm{x}}_{\bm{a}})\right] \otimes \sigma_{B'} \\ 
& =A^{\bm{x}}_{\bm{a}} \otimes \sigma_{B'},
\end{align}
where we used the notation $\sigma_{B'} = \tr_{A'P} \ketbraq{\junk}_{A'B'P}$, for the second equality we used \eqref{uniformMarg}, for the third equality we used $V_{\tilde{A}}^\dagger V_{\tilde{A}}=\mathbb{I}_{\tA}$, for the fourth equality we used the cyclicity of the partial trace, for the fifth equality we used the self-testing statements \eqref{eq:CHSHN_selftest_sumab}, for the sixth equality we used the property  $(A^{\bm{x}}_{\bm{a}} \otimes \mathbb{I}) \ket{\phi^+}^{\otimes N}=\left(\mathbb{I} \otimes  (A^{\bm{x}}_{\bm{a}})^T\right) \ket{\phi^+}^{\otimes N}$ and omitted the transpose since $A^{\bm{x}}_{\bm{a}}$ are real \eqref{CHSHNmeas}, and for the seventh equality we used $\tr_A(\ketbraq{\phi^+}^{\otimes N})=\frac{\mathbb{I}}{2^N}$. Notice that \eqref{uniformMarg} and \eqref{ImpImplication} only involve Alice's marginal probabilities and measurement operators.

Let us now consider imperfect detectors in $B_1$ which sometimes fail to click resulting in a no-click event $\emptyset$, that is, the outcome set of $B_1$ is $\bm{b}\in\{0,1\}^N\cup\{\emptyset\}$. We assume that the detection efficiency is the same for all measurements in $B_1$ and therefore $p(\bm{b}=\emptyset|\bm{y})=\eta$. We first consider the $r\text{BB84}^N$ strategy followed by $r\text{CHSH}^N$.

\subsection{$\eta^*=1/2^N$ with $r\text{BB84}^N$}
To demonstrate that the critical detection efficiency of $B_1$ for the certification of NJM measurements in $B_1$ with the $r\text{BB84}^N$ strategy is $1/2^N$, we consider a penalized version of the $N$-product BB84 inequality \eqref{BB84Ngame} for $A$ and $B_1$,
\begin{equation}\label{BB84Nqgame}
    \mathcal{B}^{ N}_{AB_1}(q)= \mathcal{B}^{ N}_{AB_1}-\frac{q}{2^N}\sum_{\bm{y},\bm{b}\in\{0,1\}^N}p_{B}(\bm{b}|\bm{y},1),
\end{equation}
where $q\in[0,1]$ is the penalty parameter. This inequality is a generalization of the long-path inequality in Eq. $(63)$ of \cite{Lobo2024certifyinglongrange}. Effectively, $A$ and $B_1$ are penalized whenever $B_1$ clicks, that is, whenever $\bm{b}\neq\emptyset$. The $r\text{BB84}^N$ strategy with $B_1$ clicking with efficiency $\eta$ achieves the following value of the $q$-penalized $N$-product BB84 inequality \eqref{BB84Nqgame},
\begin{equation} \label{BB84idealScore}
    \mathcal{B}^{ N}_{AB_1}(q)=(1-q)\eta.
\end{equation}

Next, via the following proposition, we derive an upper bound on the maximal value of $\mathcal{B}^{N}_{AB_1}(q)$ achievable with JM measurements in $B_1$ to compare against \eqref{BB84idealScore}. 

\begin{prop} \label{BB84prop}
When the devices $A$ and $B_0$ witness the maximal quantum violation of the $N$-product CHSH inequality \eqref{CHSHNTB}, the maximal achievable value of the $q$-penalized $N$-product BB84 inequality \eqref{BB84Nqgame} with jointly-measurable (JM) measurements in $B_1$ satisfies
\begin{equation} \label{BB84maxJM}
    \mathcal{B}^{ N}_{AB_1}(q)\underset{\text{JM}({B_1})}{\leq} \frac{1-q}{2^N},
\end{equation}
for $q\in[1/\sqrt{2},1]$.
\end{prop}
\begin{proof}
Let us rewrite the value of the $q$-penalized $N$-product BB84 inequality \eqref{BB84Nqgame} with JM measurements for $B_1$ using \eqref{JM} as
\begin{align}\label{BB84Nqgame1JM} 
    &\mathcal{B}^{ N}_{AB_1}(q)= \frac{1}{2^N}\sum_{\bm{y},\bm{b}\in\{0,1\}^N}\left[p(\bm{b},\bm{b}|\bm{y},\bm{y},1) -qp_{B}(\bm{b}|\bm{y},1)\right] , \\ \nonumber
    &\underset{\text{JM}(B_1)}{=}\frac{1}{2^N}\sum_{\bm{y},\bm{b}\in \{0,1\}^N} \sum_{\vec{\bm{b}}} \delta_{ \bm{b}_{\bm{y}} , \bm{b}} \Big[ \tr(\tilde{\rho}_{\tA \tB} \tilde{A}^{\bm{y}}_{\bm{b}} \otimes 
 E_{ \vec{\bm{b}} })  \\ \nonumber & \hspace{180pt}
    -q\tr(\tilde{\rho}_{\tB} E_{\vec{\bm{b}}})\Big], 
\end{align}
where the POVM $\{E_{\vec{\bm{b}}}\}_{\vec{\bm{b}}}$ describes
the joint measurement with outcomes $\vec{\bm{b}} \equiv ( \bm{b}_{ \bm{y} } \in \{0,1\}^N \cup \{\emptyset\} )_{\bm{y}}$, $\tilde{\rho}_{\tA \tB}$ and $\{\tilde{A}^{\bm{x}}_{\bm{a}}\}$ are the hereto unknown shared state and Alice's measurement operators, and $\tilde{\rho}_{\tB}=\tr_A(\tilde{\rho}_{\tA \tB})$ describes Bob's marginal state. 

Next, we express \eqref{BB84Nqgame1JM} in terms of Alice's marginal probabilities $p_A( \bm{a}|\bm{x} ) = \tr( \tilde{\rho}_{\tA \tB} \tilde{A}^{ \bm{x} }_{ \bm{a} } \otimes \mathbb{I}_{\tB} )$ of obtaining the outcome $\bm{a}$ given input $\bm{x}$ and the corresponding remotely prepared states $\tilde{ \rho }^{(\tilde{B})}_{ \bm{a} | \bm{x} } = \frac{ \tr_A( \tilde{\rho}_{\tA \tB} \tilde{A}^{ \bm{x} }_{ \bm{a} } \otimes \mathbb{I}_{\tB} ) }{ p_A(\bm{a} | \bm{x}) }$ as
\begin{align} \label{BB84Nqgame2JM}
    \mathcal{B}^{ N}_{AB_1}(q) = \frac{1}{2^N} \sum_{ \vec{\bm{b}} } \tr\left[ E_{ \vec{\bm{b}} } \sum_{ \bm{y}|\bm{b}_{\bm{y}} \neq  \emptyset } \left( p_A( \bm{b}_{\bm{y}}|\bm{y} ) \tilde{\rho}^{(\tilde{B})}_{ \bm{b}_{\bm{y}}|\bm{y} } - q\tilde{\rho}_{\tB} \right)
    \right].
\end{align}
We now use the self-testing implications \eqref{uniformMarg} and \eqref{ImpImplication} from the maximum violation of the $N$-product CHSH inequality between $A$ and $B_0$ to obtain,
\begin{equation}\label{BB84Nqgame3JM}
\begin{split}
    \mathcal{B}^{ N}_{AB_1}(q) & =  \left. \frac{1}{2^N} \sum_{ \vec{\bm{b}} } \tr\Big[ V_{\tilde{B}}^\dagger V_{\tilde{B}} E_{ \vec{\bm{b}} } V_{\tilde{B}}^\dagger V_{\tilde{B}} \right. \\
    & \ \ \ \ \ \ \ \ \ \ \ \ \ \ \ \ \left. \sum_{ \bm{y}|\bm{b}_{\bm{y}} \neq  \emptyset } \left( p_A( \bm{b}_{\bm{y}}|\bm{y} ) \tilde{\rho}^{(\tilde{B})}_{ \bm{b}_{\bm{y}}|\bm{y} } - q\tilde{\rho}_{\tB} \right)
    \Big]. \right. \\
    & \left. = \frac{1}{2^{2N}} \sum_{ \vec{\bm{b}} } \tr \left[ E'_{ \vec{\bm{b}} } \sum_{ \bm{y}|\bm{b}_{\bm{y}}\neq \emptyset } \left( A^{\bm{y}}_{\bm{b}_{\bm{y}}} - q\mathbb{I}_B \right) \otimes \sigma_{B'} \right],
    \right.
\end{split}
\end{equation}
where we used Eq.~\eqref{ImpImplication} and consequently that $V_{\tilde{B}} \tilde{\rho}_{B} V_{\tilde{B}}^\dagger = \sum_{ \bm{a} } p_A ( \bm{a}|\bm{x} ) A^{\bm{x}}_{\bm{a}} \otimes \sigma_{B'} = \frac{ \mathbb{I}_B }{ 2^N } \otimes \sigma_{B'}$, and we introduced operators $\{E'_{ \vec{\bm{b}} } = V_{\tilde{B}} E_{ \vec{\bm{b}} } V_{\tilde{B}}^\dagger \succeq 0 \}_{ \vec{\bm{b}}}$ which satisfy $\sum_{\vec{\bm{b}}} E'_{\vec{\bm{b}}} = V_{\tilde{B}} V_{\tilde{B}}^\dagger \preceq \mathbb{I}_{BB'}$.

Our task is to upper bound the maximum achievable value of the linear functional $\mathcal{B}^{ N}_{AB_1}(q)$ \eqref{BB84Nqgame3JM} over operators $\{E'_{\vec{\bm{b}}}\succeq0\}_{\vec{\bm{b}}}$ such that $\sum_{\vec{\bm{b}}}E'_{\vec{\bm{b}}} \preceq \mathbb{I}_{BB'}$. To this end, let us define the (Hermitian) operators $C_{\vec{\bm{b}}}$ on $\cH_B$ via
\begin{align} \label{CoperatorsMain}
    C_{\vec{\bm{b}}} \equiv \sum_{ \bm{y}|\bm{b}_{\bm{y}}\neq \emptyset } \left( A^{\bm{y}}_{\bm{b}_{\bm{y}}} - q\mathbb{I}_B \right) =\sum_{ \bm{y}|\bm{b}_{\bm{y}}\neq \emptyset } \left( \bigotimes^N_{j=1}A^{y_j}_{b_{\bm{y},j}}  - q\mathbb{I}_B \right),
\end{align}
where we used the product structure of Alice's optimal measurements from Eq.~\eqref{CHSHNmeas}. Assume that we can bound the maximal eigenvalue of $C_{\vec{\bm{b}}}$ by some $\gamma \in \mathbb{R}$ uniformly for every $\vec{\bm{b}}$. Then it is easy to see that $\tr[ E'_{\vec{\bm{b}}}(C_{\vec{\bm{b}}}\otimes \sigma_{B'}) ] \le \gamma \tr [ E'_{\vec{\bm{b}}}( \mathbb{I}_B \otimes \sigma_{B'}) ]$, since $E'_{\vec{\bm{b}}} \succeq 0$. Therefore, we can upper bound \eqref{BB84Nqgame3JM} via
\begin{equation}\label{BB84Nqgame4JM}
\begin{split}
\mathcal{B}^{ N}_{AB_1}(q) & \left. = \frac{1}{2^{2N}} \sum_{\vec{\bm{b}}} \tr\left[ E'_{\vec{\bm{b}}}(C_{\vec{\bm{b}}}\otimes \sigma_{B'})
    \right] \right. \\ 
    & \left. \le \frac{1}{2^{2N}} \sum_{\vec{\bm{b}}} \gamma \tr [ E'_{\vec{\bm{b}}}( \mathbb{I}_B \otimes \sigma_{B'}) ] \right. \\
    & \left. \le \frac{1}{2^{2N}} \gamma \tr ( \mathbb{I}_B \otimes \sigma_{B'}) = \frac{\gamma}{ 2^N }
\right. 
\end{split}
\end{equation}
where we used that $\sum_{\vec{\bm{b}}} E'_{\vec{\bm{b}}} \preceq \mathbb{I}_{BB'}$ and $\tr(\mathbb{I}_B) = 2^N$.
Hence, our task now is to upper bound the maximum eigenvalue of the operators $\{C_{\vec{\bm{b}}}\}_{\vec{\bm{b}}}$. Let us denote a generic $C_{\vec{\bm{b}}}$ with $k$ click-events (i.e.~there are $k$ entries of $\vec{\bm{b}}$ that are not equal to $\emptyset$) by $M_k$. Explicitly, we have
\begin{equation} 
    M_k = \sum_{\mathbf{y}|\mathbf{b}_{\mathbf{y}}\neq \emptyset}\left(\bigotimes^N_{j=1} \pi_{\bm{y},j} - q \mathbb{I}_B \right),
\end{equation}
where the summation is over $k$ terms and $\pi_{\bm{y},j} \in 
\left\{\ketbra{0}{0}, \ketbra{1}{1}, \ketbra{+}{+}, \ketbra{-}{-} \right\}$, i.e.~the $\pi_{\bm{y},j}$ are Alice's measurement operators for setting ${\bm{y}}$ and some (unspecified) outcome.
The maximum eigenvalue of $M_k$ is given by
\begin{equation}\label{eq:lambdamax_Ak}
\lambda_{\max}{(M_k)} = \norm{ \sum_{\mathbf{y}|\mathbf{b}_{\mathbf{y}}\neq \emptyset}\bigotimes^N_{j=1} \pi_{\bm{y},j} } - qk,
\end{equation} 
where $\norm{.}$ is the operator norm and we used the fact that $\lambda_{\max} ({ \sum_{\mathbf{y}|\mathbf{b}_{\mathbf{y}}\neq \emptyset}\bigotimes^N_{j=1} \pi_{\bm{y},j} }) = \norm{ \sum_{\mathbf{y}|\mathbf{b}_{\mathbf{y}}\neq \emptyset}\bigotimes^N_{j=1} \pi_{\bm{y},j} }$ (since $\sum_{\mathbf{y}|\mathbf{b}_{\mathbf{y}}\neq \emptyset}\bigotimes^N_{j=1} \pi_{\bm{y},j} 	\succeq0$), 
and that the summation includes $k$ terms. Let us introduce the simplified notation
\begin{equation}
S_{l} \equiv \bigotimes^N_{j=1} \pi_{l,j} ,
\end{equation}
where $l \in \{1, \ldots, k\}$ denotes the $\bm{y}$ in the summation in Eq.~\eqref{eq:lambdamax_Ak}. Note that for a fixed operator $M_k$, $S_l \neq S_{l'}$ if $l \neq l'$ (because different $S_l$ correspond to different measurement settings, and therefore at least one $\pi_{l,j} \neq \pi_{l',j}$). 

In the following, we fix $q = \frac{1}{\sqrt{2}}$. It is easy to compute the norm of 0-click and 1-click operators:
\begin{equation}
\norm{M_0} = 0
\end{equation}
\begin{equation}
\norm{M_1} = \norm{ S_1 } -\frac{1}{\sqrt{2}} = 1 -\frac{1}{\sqrt{2}} > \norm{M_0},
\end{equation}
irrespective of $\bm{y}$, since $\norm{ S_l } = 1$ as these are projections. In the following, we show that $\norm{M_k} \le \norm{M_1}$ for all possible settings, and therefore $\norm{M_1}=1 -\frac{1}{\sqrt{2}}$ upper bounds the maximal eigenvalue of any possible $C_{\vec{\bm{b}}}$. 

The norm of a generic $M_k$ is written as
\begin{equation}
\norm{M_k} = \norm{ S_1 + S_2 + \cdots S_k} - \frac{k}{\sqrt{2}},
\end{equation}
and therefore we need to bound $\norm{ S_1 + S_2 + \cdots S_k}$. For this, we use a result by Popovici and Sebesty\'en \cite{SEBESTYÉN_POPOVICI_2005}: For any set of positive semidefinite matrices (such as the $S_l$), we have
\begin{equation}
\norm{ S_1 + S_2 + \cdots S_k} \le \norm{\Gamma},
\end{equation}
where $\Gamma$ is a $k$-by-$k$ matrix with elements $\Gamma_{ll'} = \norm{\sqrt{S_l}\sqrt{S_{l'}}}$. For our case, we have $\sqrt{S_l} = S_l$, since $S_l$ are projections, which also means that $\Gamma_{ll} = 1$ for all $l$. The rest of the elements are given by
\begin{equation}
\norm{ S_l S_{l'} } = \norm{ \bigotimes_{k=1}^N \pi_{l,k} \pi_{l',k} } \le \prod_{j=1}^{N} \norm{\pi_{l,j} \pi_{l',j} }
\end{equation}
by the submultiplicativity of the operator norm.
Note that $S_l$ and $S_{l'}$ differ in the type of Pauli on at least one tensor factor $j^\ast$, and for this particular factor we therefore have $\norm{\pi_{l,j^\ast} \pi_{l',j^\ast} } = \frac{1}{\sqrt{2}}$. The rest of the norms are bounded by $1$, and therefore in general we have
\begin{equation}
    \Gamma_{l,l'} \leq\begin{cases}1, \ & \text{if} \ l=l', \\
    \frac{1}{\sqrt{2}}, \  & \text{if} \ l\neq l'.
    \end{cases}
\end{equation}
This implies that the matrix $\Gamma$ is upper bounded element-wise by the matrix
\begin{align}
G_k &= \begin{pmatrix}
1 & \frac{1}{\sqrt{2}} & \frac{1}{\sqrt{2}} & \cdots & \frac{1}{\sqrt{2}} \\
\frac{1}{\sqrt{2}} & 1 & \frac{1}{\sqrt{2}} & \cdots & \frac{1}{\sqrt{2}} \\
\vdots & & \ddots & & \vdots \\
\frac{1}{\sqrt{2}} & \frac{1}{\sqrt{2}} & \frac{1}{\sqrt{2}} & \cdots & 1
\end{pmatrix} \\ \nonumber &
= \left( 1 - \frac{1}{\sqrt{2}} \right) \mathbb{I} + \begin{pmatrix}
\frac{1}{\sqrt{2}} & \frac{1}{\sqrt{2}} & \frac{1}{\sqrt{2}} & \cdots & \frac{1}{\sqrt{2}} \\
\frac{1}{\sqrt{2}} & \frac{1}{\sqrt{2}} & \frac{1}{\sqrt{2}} & \cdots & \frac{1}{\sqrt{2}} \\
\vdots & & \ddots & & \vdots \\
\frac{1}{\sqrt{2}} & \frac{1}{\sqrt{2}} & \frac{1}{\sqrt{2}} & \cdots & \frac{1}{\sqrt{2}}
\end{pmatrix}.
\end{align}
Since every element of $\Gamma$ and $G_k$ is non-negative, the element-wise inequality implies $\norm{\Gamma} \le \norm{G_k}$ (see Appendix \ref{app:elementwise}), and therefore
\begin{equation}
\norm{\Gamma}\leq \norm{G_k} = \left( 1 - \frac{1}{\sqrt{2}} \right) + \frac{k}{\sqrt{2}} = 1 + \frac{k-1}{\sqrt{2}}.
\end{equation}
putting everything together, we obtain
\begin{equation}
\norm{M_k} \le 1 + \frac{k-1}{\sqrt{2}} - \frac{k}{\sqrt{2}} = 1 - \frac{1}{\sqrt{2}} = \norm{M_1},
\end{equation}
as desired. The same argument works for every $\frac{1}{\sqrt{2}} \le q < 1$, and overall we get
\begin{equation}
    \norm{C_{\bm{\vec{b}}}} \leq 1-q=\gamma,
\end{equation}
which when plugged back into \eqref{BB84Nqgame4JM} yields \eqref{BB84maxJM}.


\end{proof}

As a direct consequence of Proposition \ref{BB84prop}, we find that the $r\text{BB84}^N$ strategy \eqref{BB84idealScore} violates the JM threshold \eqref{BB84maxJM} for all $\eta>\eta^*=1/2^N$. Next, we consider the $r\text{CHSH}^N$ strategy.

\subsection{$\eta^*=1/2^N$ with $r\text{CHSH}^N$}
Proceeding in a similar way, we consider a penalized version of the $N$-product CHSH inequality \eqref{CHSHNTB} for $A$ and $B_1$,
\begin{equation}\label{CHSHNq}
    \mathcal{C}^{N}_{AB_1}(q)= \mathcal{C}^{N}_{AB_1}-\frac{q}{2^N}\sum_{\bm{y},\bm{b}\in\{0,1\}^N}p_{B}(\bm{b}|\bm{y},1),
\end{equation}
where $q\in[0,\frac{1}{2^N}(1+1/\sqrt{2})^N)$. A straightforward computation shows that the $r\text{CHSH}^N$ strategy with $B_1$ clicking with efficiency $\eta$ achieves 
\begin{equation} \label{CHSHidealScore}
    \mathcal{C}^{N}_{AB_1}(q)=\left(\alpha^N-q\right)\eta.
\end{equation}
Next, via the following proposition we derive an upper bound on the maximal value of \eqref{CHSHNq} achievable with JM measurements in $B_1$ to compare with \eqref{CHSHidealScore}.
\begin{prop} \label{CHSHProp}
    When the devices $A$ and $B_0$ witness the maximal quantum violation of the $N$-product CHSH inequality \eqref{CHSHNTB}, the maximal achievable value of the $q$-penalized $N$-product CHSH inequality \eqref{CHSHNq} with joint-measurable (JM) measurements in $B_1$ satifies

\begin{equation} \label{CHSHmaxJM}
    \mathcal{C}^{ N}_{AB_1}(q)\underset{\mathcal{JM}(B_1)}\leq \frac{1}{2^N}\left(\alpha^N-q\right),
\end{equation}
for $q\in[{\alpha^{N-1}\beta}, {\alpha^N}]$ where $\beta=\frac{1}{8} \left(2+\sqrt{2}+\sqrt{4 \sqrt{2}-2}\right)\approx 0.66581$. 
\end{prop} 
\begin{proof}
The proof proceeds in a similar way to the proof of Proposition \ref{BB84prop} and has been deferred to the appendix for brevity. Notably, unlike the case of the $r\text{BB84}^{ N}$ strategies, where each $C_{\vec{\bm{b}}}$ contains a sum of tensor products of two-dimensional projections $\pi_{\bm{y},j}\in\{\ketbra{0}{0},\ketbra{1}{1},\ketbra{+}{+},\ketbra{-}{-}\}$ \eqref{BB84Nqgame4JM}, for $r\text{CHSH}^{ N}$ strategies each $C_{\vec{\bm{b}}}$ contains a sum of tensor products of two-dimensional density operators $\pi_{\mathbf{y},i}\in\{\frac{1}{2}(\ketbra{0}{0}+\ketbra{+}{+}),\frac{1}{2}(\ketbra{0}{0}+\ketbra{-}{-}),\frac{1}{2}(\ketbra{1}{1}+\ketbra{+}{+}),\frac{1}{2}(\ketbra{1}{1}+\ketbra{-}{-})\}$. This results in the Popovici--Sebesty\'en bound \cite{SEBESTYÉN_POPOVICI_2005}
not being tight, also leading to the lower bound on $q$ [$q\geq \alpha^{N-1}\beta$, for \eqref{CHSHmaxJM}] not being tight. An exhaustive search over the ${(2^N+1)}^{2^N}$ operators $\{C_{\vec{\bm{b}}}\}_{\vec{\bm{b}}}$ for $N=1,2,3$ yields a tighter lower bound $q\geq \alpha^{N-1}\beta'$  \eqref{CHSHmaxJM}, with $\beta'=\frac{4 - \sqrt{2}}{4}\approx 0.64644\leq \beta$, which we conjecture to also hold for $N>3$ (code available online \cite{chaturvedi2025parrelreproutedbell}).
\end{proof}
Proposition \ref{CHSHProp} implies that for $q\in[{\alpha^{N-1}\beta}, {\alpha^N})$, the $r\text{CHSH}^N$ strategy achieving \eqref{CHSHidealScore} violates the JM threshold \eqref{CHSHmaxJM} whenever $\eta>\eta^*=1/2^N$. 

Until now, we have assumed that that the devices $A$ and  $B_0$ witness perfect nonlocal correlations which maximally violate the $N$-product CHSH inequality. In the next section, we analyze the robustness of our results.

\section{Robustness}
In this section, we demonstrate the robustness of our results in two different ways.  First, we recall that finding the critical detection efficiency $\eta^*$ of $B_1$ for the certification of NJM measurements in $B_1$ given that the devices $A$ and $B_0$ witness imperfect nonlocal correlations, forms a noncommutative polynomial optimization (NPO) problem \cite{Navascues2007,Navas_2008,PNA10}. The optimization problem is further simplified by taking advantage of the fact that for projective measurements, joint measurability is equivalent to the commutativity of the corresponding projection operators, and all measurements in routed Bell experiments can be taken as projective. Thus, instead of enforcing the condition \eqref{JM} for jointly measurable measurements in $B_1$, commutation relations $[B^{\bm{y},1}_{\bm{b}},B^{\bm{y}',1}_{\bm{b}'}]=0$ for all $\bm{y},\bm{y}'\in\{0,1\}^N$ are imposed \cite{Lobo2024certifyinglongrange}. The resultant NPO problem can then be relaxed and the relaxations solved via the Navascu\'es--Pironio--Ac\'in (NPA) hierarchy \cite{Navascues2007,Navas_2008,PNA10}, thereby providing upper bounds on the critical detection efficiency $\eta^*$.
We employ this method to compare how the single-copy strategies $r\text{CHSH}$ and $r\text{BB84}$ and their parallel repeated versions $r\text{CHSH}^2$ and $r\text{BB84}^2$ fare under source imperfections. Specifically, we model an imperfect source with visibility $v\in[0,1]$ that distributes the state, 
\begin{equation}
    v\ketbra{\phi^+}{\phi^+}^{\otimes N}+(1-v)\frac{\mathbb{I}_{AB}}{2^{2N}},
\end{equation}
 where $\mathbb{I}_{AB}$ is the identity operator on $\mathbb{C}^{2^{2N}}$. 
 
Keeping the entire resultant correlation between $A$ and $B_0$ as a constraint, we use the third-order relaxation of the NPA hierarchy to obtain bounds on $\eta^*$ for $r\text{CHSH}$ and $r\text{BB84}$ across $v\in[0,1]$. For $r\text{CHSH}^2$ and $r\text{BB84}^2$, we use the $1+AB$-order relaxation. We plot these results in Fig. \ref{EtaVsVisibility}. We find that both parallel repeated strategies exhibit considerably improved robustness to source imperfections compared to their single-copy counterparts. The python codes are available online \cite{chaturvedi2025parrelreproutedbell}.
\begin{figure}[ht]
    \centering
    \includegraphics[width=\linewidth]{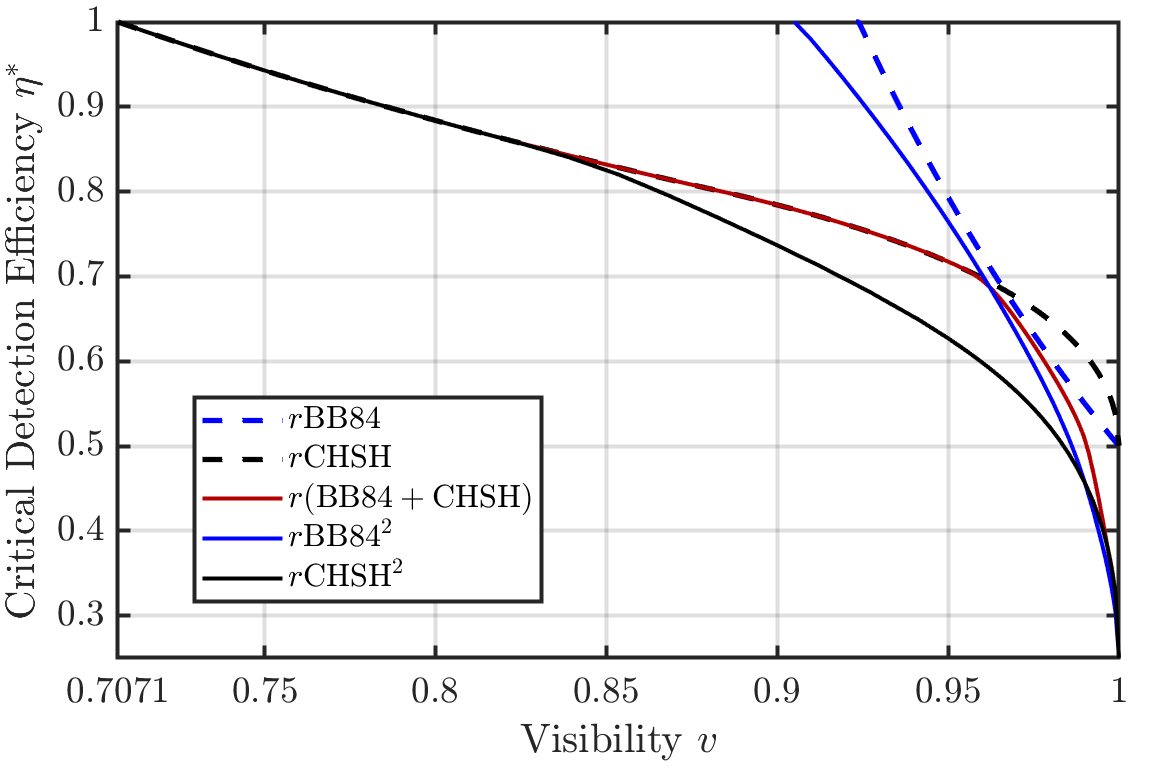}
    \caption{Comparison of the upper bounds on the critical detection efficiency $\eta^*$ obtained via the NPA hierarchy for the certification of NJM measurements in device $B_1$ as a function of source visibility $v$  \cite{chaturvedi2025parrelreproutedbell}. The devices $A,B_0$ are assumed to have perfect detection. The figure compares the single-copy strategies $r\text{BB84}$ and $r\text{CHSH}$ (dashed blue and black curves, respectively, obtained at level $3$ of the NPA hierarchy \cite{Navascues2007,Navas_2008,PNA10}) to their parallel repeated counterparts $r\text{BB84}^2$ and $r\text{CHSH}^2$ (solid blue and black curves, respectively, obtained at level $1+AB$). We find that both parallel repeated strategies exhibit considerably improved robustness to source imperfections compared to their single-copy counterparts. 
    The solid red curve obtained at level $2+AAB_1$ corresponds to the qubit-based $r(\text{BB84}+\text{CHSH})$ \eqref{r(BB84+CHSH)} from \cite{sekatski2025certificationquantumcorrelationsdiqkd}, which is also surpassed by the ququart-based $r\text{BB84}^2$ and $r\text{CHSH}^2$ strategies. }
    \label{EtaVsVisibility}
\end{figure}

This optimization problem is computationally highly demanding and becomes impractical for $N>2$. To demonstrate the in-principle robustness of our results for $N>2$, we translate the robustness of self-testing statements for the maximal violation of $N$-product CHSH inequalities to bounds on the critical detection efficiency in the routed scenario. Let us suppose that we have a robust self-testing statement for the $N$-product CHSH inequalities, parametrized by some $\varepsilon > 0$, the gap between the maximal and actual Bell inequality violation. In particular, assume that we have that for all purifications $\ket{\tpsi}_{\tA \tB P}$ of all physical states $\tilde{\rho}_{\tA \tB}$ on $\cH_{\tA} \otimes \cH_{\tB}$ and physical measurements $\{\tA^{\bm{x}}_{\bm{a}}\}_{\bm{a}}$ on $\mathcal{H}_{\tA}$ and $\{\tB^{\bm{y}}_{\bm{b}}\}_{\bm{b}}$ on $\mathcal{H}_{\tB}$ that are $\varepsilon$-compatible with the target distribution, there exist Hilbert spaces $\cH_{A'}$ and $\cH_{B'}$ and isometries $V_{\tilde{A}} : \mathcal{H}_{\tA} \to \mathbb{C}^{2^N} \otimes \mathcal{H}_{A'}$ and $V_{\tilde{B}}: \mathcal{H}_B \to \mathbb{C}^{2^N} \otimes \mathcal{H}_{B'}$ such that
\begin{equation} \label{robustSelfTestCHSHN}
\begin{aligned}
&\Bigl\| 
\bigl(V_{\tilde{A}} \otimes V_{\tilde{B}} \otimes \mathbb{I}_P\bigr)
\bigl(\tilde{A}^{\bm{x}}_{\bm{a}} \otimes \tilde{B}^{\bm{y}}_{\bm{b}} \otimes \mathbb{I}_P\bigr)\ket{\tilde{\psi}}_{ABP} \\
&\quad - \bigl({A}^{\bm{x}}_{\bm{a}} \otimes {B}^{\bm{y}}_{\bm{b}}\bigr)\ket{\phi^+}^{\otimes N}
\otimes \ket{\text{junk}}_{A'B'P}
\Bigr\|
\;\le\; f(N,\varepsilon)
\end{aligned}
\end{equation}
for some function $f(N,\varepsilon) \ge 0$ such that $f(N,\varepsilon) \to 0$ as $\varepsilon \to 0$.
Then via the following proposition we demonstrate that the robust self-testing statement \eqref{robustSelfTestCHSHN} translates to the robustness of the critical detection efficiency in the $q$-penalized $N$-product BB84 inequality.
\begin{prop}
    Suppose that $A$ and $B_0$ witness $\mathcal{C}^N_{AB_0}=\alpha^N-\varepsilon$, for some $\varepsilon > 0$. Then, given a robust self-testing statement of the form \eqref{robustSelfTestCHSHN}, there exists $0 < \delta \le 2^{2N+1} \left[ f^2(N,\varepsilon) + 2f(N,\varepsilon) \right] + 2^{4N} \left[ f^2(N,\varepsilon) + 2f(N,\varepsilon) \right]^2$ such that by setting $q = \left( \frac{1}{\sqrt{2}} + \sqrt{\delta} \right)$, the critical detection efficiency for the $q$-penalized $N$-product BB84 inequality \eqref{BB84Nqgame} in $B_1$ is given by
\begin{equation} \label{BB84maxJMRobust}
    \eta^* = \frac{1}{2^N} \frac{ 1 - \frac{1}{\sqrt{2}} }{ 1 - \frac{1}{\sqrt{2}} - \sqrt{ \delta } }
\end{equation}
assuming the only imperfection is the detection efficiency. Assuming that the imperfection $\varepsilon$ translates linearly from the CHSH setting to the BB84 setting (e.g.~$\varepsilon$ is the result of white noise on the state), the critical detection efficiency changes to
\begin{equation} \label{BB84maxJMRobust1}
    \eta^* = \frac{1}{2^N} \frac{ 1 - \frac{1}{\sqrt{2}} }{ 1 - \frac{1}{\sqrt{2}} - \frac{\varepsilon}{\alpha^N} - \sqrt{ \delta } }.
\end{equation}
\end{prop}
\begin{proof}
    The proof utilizes the Popovici--Sebesty\'en bound \cite{SEBESTYÉN_POPOVICI_2005} and has been deferred to Appendix \ref{app:robustness} for brevity. 
\end{proof}
We expect a similar robustness statement to hold for the $q$-penalized $N$-product CHSH inequality, but the details are left for a future study.

\section{Conclusions}
The combined effect of inevitable transmission losses and the limited efficiencies of detectors constitutes one of the most persistent challenges in the practical implementation of device-independent (DI) quantum cryptography based on Bell nonlocality. To address this issue, we integrate two of the most prominent recently proposed approaches: parallel repetition \cite{miklin2022exponentially} and routed Bell experiments \cite{chaturvedi2024extending}. Specifically, we consider the $N$-fold parallel repeated versions of the simple $r\text{BB84}$ and $r\text{CHSH}$ strategies, termed $r\text{BB84}^N$ and $r\text{CHSH}^N$ for routed Bell experiments. We demonstrate that when the source and the devices close to the source $A$ and $B_0$ are perfect, the critical detection efficiency required for the certification of NJM measurements in $B_1$ decreases exponentially following $\eta^*=\frac{1}{2^N}$. Since the number of inputs for $B_1$ in $r\text{BB84}^N$ and $r\text{CHSH}^N$ is $2^N$, the threshold $\eta^*=\frac{1}{2^N}$ is tight. 

We demonstrate the robustness of these results in two ways. First, we analyze a realistic configuration where devices $A$ and $B_0$, along with the switch, reside in a ``central hub" equipped with state-of-the-art technology, allowing us to assume near-perfect performance for these components. Under this assumption, we observe that the parallel repeated strategies $r\text{BB84}^2$ and $r\text{CHSH}^2$ exhibit significantly enhanced robustness to imperfections in the source compared to their single-copy counterparts $r\text{BB84}$ and $r\text{CHSH}$. Second, we provide an analytical argument showing that the robustness of self-testing statements for the $N$-product CHSH inequality carries over to the robustness of our results for the $r\text{BB84}^N$ strategy, which we also expect to hold for the $r\text{CHSH}^N$ strategy.  

A remarkable feat of the $r\text{BB84}$ and $r\text{CHSH}$ strategies is that they enable secure DIQKD at $\eta^*\approx 1/2$ for $B_1$ \cite{roydeloison2024deviceindependent}. The $r\text{BB84}^N$ and $r\text{CHSH}^N$ strategies, therefore, have the potential to facilitate  noise-robust DIQKD with $\eta^*\approx 1/2^N$ for $B_1$ and a higher key-rate because of more number of outcomes. Another intriguing direction is exploring whether other self-testing statements can be utilized to find better strategies for routed Bell experiments. Our proof technique could help us answer this question.  


\section*{Note Added}
While finalizing this manuscript, we became aware of a recent preprint~\cite{sekatski2025certificationquantumcorrelationsdiqkd}. That work demonstrates that if devices $A$ and $B_0$ achieve the maximal quantum violation of the CHSH inequality \eqref{TsirelsonsB} using the strategy \eqref{optMeasCHSH}, and $B_1$ performs $n$ dichotomic measurements on the $X\!-\!Z$ plane of the Bloch sphere, then for $B_1$, $\eta^* = \frac{1}{n}$ for both certifying non-jointly measurable (NJM) measurements in $B_1$ and achieving secure DIQKD. Because these protocols rely only on qubits, they are notably simpler than $r\text{BB84}^N$ and $r\text{CHSH}^N$, which require local dimension $2^N$. 

We focus on a specific example of such a qubit strategy, referred to as $r(\text{BB84}+\text{CHSH})$. In this strategy, the source distributes the maximally entangled state $\ket{\phi^+}_{AB}$ and devices $A$ and $B_0$ act as in the $r\text{BB84}$ and $r\text{CHSH}$ protocols, while $B_1$ measures $\{\{B^{y,1}_b\}_{b\in\{0,1\}}\}_{y\in\{0,1,2,3\}}$—four dichotomic observables constructed by “combining” the measurements of $A$ and $B_0$. Concretely,
\begin{equation} \label{r(BB84+CHSH)}
    B^{y,1}_b 
= 
\begin{cases}
  A^{y}_b, & \text{if } y \in \{0,1\},\\
  B^{y-2}_b, & \text{if } y \in \{2,3\},
\end{cases}
\end{equation}
where $A^{\bm{x}}_{\bm{a}}$ and $B^{\bm{y}}_{\bm{b}}$ are defined in \eqref{optMeasCHSH}.

Using level-$2 + AAB_1$ of the NPA hierarchy, we compute $\eta^*$ for the certification NJM measurements in $B_1$ as a function of the source visibility $v$ (code available online \cite{chaturvedi2025parrelreproutedbell}). The results plotted in Fig. \ref{EtaVsVisibility} show that while $r(\text{BB84}+\text{CHSH})$ attains $\eta^* \approx \frac{1}{4}$ when $v=1$ (just like $r\text{BB84}^2$ and $r\text{CHSH}^2$), its performance degrades faster for $v < 1$. In other words, the higher-dimensional protocols $r\text{BB84}^2$ and $r\text{CHSH}^2$ exhibit lower critical detection efficiencies $\eta^*$ for the certification NJM measurements in $B_1$ under non-ideal visibilities, thus highlighting the advantage of higher dimensional strategies in routed Bell experiments. We view this as a strong motivation for further investigation into strategies for routed Bell experiments built on higher-dimensional self-testing statements.
\section{Acknowledgements}
We thank Edwin Peter Lobo for pointing the $r(\text{BB84}+\text{CHSH})$ strategy \eqref{r(BB84+CHSH)}. We thank Edwin Peter Lobo, Nikolai Miklin, Stefano Pironio, Giuseppe Viola, Ekta Panwar, Jan-Åke Larsson, and Marek Żukowski for insightful discussions. 
AC acknowledges financial support by NCN grant SONATINA 6 (contract No. UMO-2022/44/C/ST2/00081).  
The numerical optimization was carried out using \href{https://ncpol2sdpa.readthedocs.io/en/stable/index.html}{Ncpol2sdpa} \cite{wittek2015algorithm}, \href{https://yalmip.github.io/}{YALMIP} \cite{Lofberg2004}, and \href{https://www.mosek.com/documentation/}{MOSEK} \cite{mosek}. This work was partially supported by the Foundation for Polish Science (IRAP project, ICTQT, contract No. MAB/218/5, co-financed by EU within the Smart Growth Operational Programme).

\bibliographystyle{apsrev4-2}
\bibliography{routed_parallel_CHSH}
\clearpage
\onecolumngrid
\appendix
\section{Element-wise inequality implies operator norm inequality}\label{app:elementwise}
\begin{lem}\label{lem:elementwise}
Let $A,B \in M_n(\mathbb{C})$ such that $0 \le A_{ij} \le B_{ij}$ for all $i,j$. Then $\norm{A} \le \norm{B}$.
\end{lem}

\begin{proof}
By definition,
\begin{equation}
\norm{A} = \max_{\substack{x \in \mathbb{C}^n \\ \norm{x} = 1}} \norm{Ax} = \max_{\substack{x \in \mathbb{C}^n \\ \norm{x} = 1}} \sqrt{ \sum_i \abs{ \sum_j A_{ij} x_j }^2 }
\end{equation}

We first notice that the maximum is attained at a vector $x$ such the all its elements are real and non-negative, $x_j \ge 0$. To see this, consider a single term in the summation over $i$:
\begin{equation}
\begin{split}
\abs{ \sum_j A_{ij} x_j }^2 & \left. = \sum_{j,k} A_{ij} A_{ik} \bar{x}_j x_k = \abs{ \sum_{j,k} A_{ij} A_{ik} \bar{x}_j x_k } \le \sum_{j,k} \abs{ A_{ij} A_{ik} \bar{x}_j x_k } \right. \\
& \left. = \sum_{j,k}  A_{ij} A_{ik} \abs{ \bar{x}_j x_k } \le \sum_{j,k}  A_{ij} A_{ik} \abs{ x_j } \abs{ x_k } = \abs{ \sum_j A_{ij} \abs{x_j} }^2,
\right.
\end{split}
\end{equation}
where we used that $A_{ij} \ge 0$. Therefore, we have $\norm{Ax} \le \norm{Ax'}$, where $x'_j = \abs{x_j} \ge 0$. We can then write the norm of $A$ as
\begin{equation}
\norm{A} = \max_{\substack{x \in \mathbb{C}^n \\ \norm{x} = 1 \\ x_j \ge 0}} \sqrt{ \sum_i \left( \sum_j A_{ij} x_j \right)^2 }.
\end{equation}
Let us denote by $x^*$ a vector reaching the maximum in the above formulation of $\norm{A}$. Using the multinomial theorem, we then have
\begin{equation}
\begin{split}
\left( \sum_j A_{ij} x^*_j \right)^2 & \left. = \sum_{\substack{ k_1, k_2, \ldots, k_n \ge 0 \\ k_1 + k_2 + \cdots k_n = 2}} \binom{2}{k_1,k_2, \ldots, k_n} (A_{i1} x^*_1)^{k_1} (A_{i2} x^*_2)^{k_2} \cdots (A_{in} x^*_n)^{k_n} \right. \\
& \left. \le
\sum_{\substack{ k_1, k_2, \ldots, k_n \ge 0 \\ k_1 + k_2 + \cdots k_n = 2}} \binom{2}{k_1,k_2, \ldots, k_n} (B_{i1} x^*_1)^{k_1} (B_{i2} x^*_2)^{k_2} \cdots (B_{in} x^*_n)^{k_n} = \left( \sum_j B_{ij} x^*_j \right)^2,
\right.
\end{split}
\end{equation}
where we also used that $x^*_j \ge 0$ and $0 \le A_{ij} \le B_{ij}$. This inequality then implies that $\norm{A} \le \norm{B}$, as desired. 
\end{proof}

\section{Proof of Proposition \ref{CHSHProp}}

Let us begin by rewriting the $q$-penalized $N$-product CHSH inequality \eqref{CHSHNq} with jointly-measurable (JM) measurements for $B_1$ using \eqref{JM} as,
\begin{align}\label{CHSHNqgame1JM} 
    \mathcal{C}^{ N}_{A,B_1}(q)&= \frac{1}{2^{2N}}\sum_{\bm{y},\bm{b}\in\{0,1\}^N}\left(\sum_{\bm{x}}p(\bm{b}\oplus \bm{x}\cdot \bm{y},\bm{b}|\bm{x},\bm{y},1) -2^Nqp_{B}(\bm{b}|\bm{y},1)\right)  \\ \nonumber
    &=\frac{1}{2^{2N}}\sum_{\bm{y},\bm{b}\in\{0,1\}^N}\sum_{\vec{\bm{b}}}\delta_{\bm{b}_{\bm{y}},\bm{b}}\left\{\sum_{\bm{x}}\tr(\tilde{\rho}_{AB}\tilde{A}^{\bm{x}}_{\bm{\bm{b}\oplus \bm{x}\cdot \bm{y}}}\otimes E_{\vec{\bm{b}}}) -2^Nq\tr(\tilde{\rho}_{B} E_{\vec{\bm{b}}})\right\}
\end{align}
where $\bm{b}\oplus \bm{x}\cdot \bm{y}\equiv (b_j\oplus x_j\cdot y_j)_{j=1}^N$, the POVM $\{E_{\vec{\bm{b}}}\}$ describes
the parent measurement with outcomes $\vec{\bm{b}}\equiv(\bm{b}_{\bm{y}}\in\{0,1\}^N\cup \{\emptyset\})_{\bm{y}}$, and $\tilde{\rho}_{\tilde{A}\tilde{B} }$, $\{\tilde{A}^{\bm{x}}_{\bm{a}=\bm{b}\oplus \bm{x}\cdot \bm{y}}\}$ are the arbitrary shared state and Alice's measurement operators with $\tilde{\rho}_{\tilde{B}}=\tr_A(\tilde{\rho}_{\tilde{A}\tilde{B} })$ being Bob's marginal state. 

Next, we express \eqref{CHSHNqgame1JM} in terms of Alice's marginal probabilities $p_A(\bm{a}|\bm{x})=\tr(\tilde{\rho}_{AB}\tilde{A}^{\bm{x}}_{\bm{a}}\otimes \mathbb{I}_{\tilde{B}})$ of obtaining outcome $\bm{a}$ given input $\bm{x}$ and the corresponding remotely prepared states of Bob, $\tilde{\rho}^{(\tilde{B})}_{\bm{a}|\bm{x}}=\frac{\tr_A(\tilde{\rho}_{\tilde{A}\tilde{B}}\tilde{A}^{\bm{x}}_{\bm{a}}\otimes \mathbb{I}_{\tilde{B}})}{p_A(\bm{a}|\bm{x})}$:

\begin{align} \label{CHSHNqgame2JM} 
    \mathcal{C}^{ N}_{A,B_1}(q)=\frac{1}{2^{2N}}\sum_{\vec{\bm{b}}}\tr\left\{E_{\vec{\bm{b}}}\sum_{\bm{y}|\bm{b}_{\bm{y}}\neq \emptyset}\left(\sum_{\bm{x}}p_A(\bm{b}_{\bm{y}}\oplus \bm{x}\cdot \bm{y}|\bm{x})\tilde{\rho}^{(\tilde{B})}_{\bm{b}_{\bm{y}}\oplus \bm{x}\cdot \bm{y}|\bm{x}}-2^Nq\tilde{\rho}_{\tilde{B}}\right)
    \right\}.
\end{align}
Our aim is to upper bound the maximum value of \eqref{CHSHNqgame2JM} given that the devices $A$ and $B_0$ witness the maximal quantum violation of the $N$-product CHSH inequality \eqref{CHSHNTB}, which self-tests the optimal quantum strategy \eqref{CHSHNselfTest}. In particular, the self-testing statements \eqref{CHSHNselfTest} imply \eqref{uniformMarg} and \eqref{ImpImplication}, using which in \eqref{CHSHNqgame2JM} we obtain
\begin{align}\label{CHSHNqgame3JM} \nonumber
    \mathcal{C}^{ N}_{A,B_1}(q)&=\frac{1}{2^{2N}}\sum_{\vec{\bm{b}}}\tr\left\{E'_{\vec{\bm{b}}}\sum_{\bm{y}|\bm{b}_{\bm{y}}\neq \emptyset}\left(\frac{1}{2^N}\sum_{\bm{x}}A^{\bm{x}}_{\bm{b}_{\bm{y}}\oplus \bm{x}\cdot \bm{y}} -q\mathbb{I}_{B}\right)\otimes \sigma_{B'}
    \right\} 
\end{align}
where $E'_{\vec{\bm{b}}}$ is defined the same way as for the BB84 case. Our aim now is therefore to upper bound the maximum achievable value of the linear functional $\mathcal{C}^{ N}_{A,B_1}(q)$ over operators $\{E'_{\vec{b}} \succeq 0\}_{\vec{\bm{b}}}$ such that $\sum_{\vec{\bm{b}}}E'_{\vec{\bm{b}}} \preceq \mathbb{I}_{BB'}$. To this end, we define the (Hermitian) operators $C_{\vec{\bm{b}}}$ via
\begin{equation} \label{Coperators}
    C_{\vec{\bm{b}}}= \sum_{\bm{y}|\bm{b}_{\bm{y}}\neq \emptyset} \left( \frac{1}{2^N} \sum_{\bm{x}}A^{\bm{x}}_{\bm{b}_{\bm{y}}\oplus \bm{x}\cdot \bm{y}} -q\mathbb{I}_B \right) = \sum_{\bm{y}|\bm{b}_{\bm{y}}\neq \emptyset}\left( \frac{1}{2^N} \sum_{\bm{x}} \bigotimes^N_{j=1} A^{x_j}_{b_{\bm{y},j}\oplus x_j\cdot y_j} -q\mathbb{I}_B \right),
\end{equation}
where for the second equality we used product structure of Alice's optimal projections \eqref{CHSHNmeas}. Similarly to the case of BB84, if we can bound the maximal eigenvalue of $C_{\vec{\bm{b}}}$ by some $\gamma \in \mathbb{R}$, then we obtain
\begin{align}\label{CHSHNqgame4JM} 
\mathcal{C}^{ N}_{A,B_1}(q) = \frac{1}{2^{2N}} \sum_{\vec{\bm{b}}}\tr\left\{ E'_{\vec{\bm{b}}} (C_{\vec{\bm{b}}} \otimes \sigma_{B'})
 \right\} \leq \frac{\gamma}{2^{N}}.
\end{align}
Hence, our task now is to upper bound the maximum eigenvalue of the operators $\{C_{\vec{\bm{b}}}\}_{\vec{\bm{b}}}$ \eqref{Coperators}. 

In \eqref{Coperators}, each $C_{\vec{\bm{b}}}$ contains a sum over $\bm{y} \in \{ \bm{y} | \bm{b}_{\bm{y}} \neq \emptyset\}$ of sums over $\bm{x} \in \{0,1\}^N$ of tensor products of projections, introducing some complications compared to the $r\text{BB84}^{ N}$ strategy. We now rewrite \eqref{CHSHNqgame4JM} such that each $C_{\vec{\bm{b}}}$ contains a sum over $\bm{y} \in \{\bm{y} | \bm{b}_{\bm{y}}\neq \emptyset\}$ of tensor products of two-dimensional density operators $\pi'_{\mathbf{y},i} \in \{\frac{1}{2}(\ketbra{0}{0}+\ketbra{+}{+}),\frac{1}{2}(\ketbra{0}{0}+\ketbra{-}{-}),\frac{1}{2}(\ketbra{1}{1}+\ketbra{+}{+}),\frac{1}{2}(\ketbra{1}{1}+\ketbra{-}{-})\}$,
\begin{align}
C_{\mathbf{\vec{b}}}& = \sum_{\bm{y}|\bm{b}_{\bm{y}}\neq \emptyset} \left( \frac{1}{2^N} \sum_{\bm{x}}\bigotimes^N_{j=1} A^{x_j}_{b_{\bm{y},j}\oplus x_j\cdot y_j} - q \mathbb{I}_B\right)\\
& = \sum_{\mathbf{y}|\mathbf{b}_\mathbf{y} \neq \emptyset} \left\{\frac{1}{2^N}\sum_{x_0}\ldots\sum_{x_{N-1}} \bigotimes_{j=1}^{N-1} A^{x_j}_{b_{\bm{y},j}\oplus x_j\cdot y_j}\ \otimes \ \left(A^{x_{N}=0}_{b_{\bm{y},N}} + A^{x_N=1}_{b_{\bm{y},j}\oplus y_N}\right) - q \mathbb{I}_B\right\} \\
& = \sum_{\mathbf{y}|\mathbf{b}_\mathbf{y} \neq \emptyset} \left\{\frac{1}{2^{N-1}}\sum_{x_0}\ldots\sum_{x_{N-1}} \bigotimes_{j=1}^{N-1} A^{x_j}_{b_{\bm{y},j}\oplus x_j\cdot y_j}\ \otimes \ \frac{1}{2} \left(\frac{\mathbb{I}_2+(-1)^{b_{\bm{y},N}}Z}{2}
+
\frac{\mathbb{I}_2+(-1)^{b_{\bm{y},N}\oplus y_N}X)}{2}\right) - q \mathbb{I}_B\right\} \\ 
&=\sum_{\mathbf{y}|\mathbf{b}_\mathbf{y} \neq \emptyset} \left( \bigotimes_{j=1}^N \pi'_{\mathbf{y},j} - q \mathbb{I}_B\right) 
\end{align}
Similarly to the BB84 case, let us denote a generic $C_{\vec{\bm{b}}}$ with $k$ click-events by $M_k$:
\begin{equation}
 M_k=  \sum_{\mathbf{y}|\mathbf{b}_\mathbf{y} \neq \emptyset} \left( \bigotimes_{j=1}^N \pi'_{\mathbf{y},j} - q \mathbb{I}\right) 
\end{equation}
Again, the maximum eigenvalue of $M_k$ is given by
\begin{equation}
\lambda_{max}(M_k) = \norm{ \sum_{\mathbf{y}|\mathbf{b}_{\mathbf{y}}\neq \emptyset}\bigotimes^N_{j=1} \pi'_{\bm{y},j} } - qk.
\end{equation}
Using the simplified notation
\begin{equation}
S_{l} \equiv \bigotimes^N_{j=1} \pi'_{l,j} ,
\end{equation}
we again have that for a fixed operator $M_k$, $S_l \neq S_{l'}$ if $l \neq l'$ (because different $S_l$ correspond to different measurement settings, and therefore at least one $\pi'_{l,j} \neq \pi'_{l',j}$).

The norm of the 0-click and 1-click operators turns out to be,
\begin{equation}
    \norm{M_0}=0
\end{equation}
\begin{align} \label{1clicknorm} 
\norm{M_1}=\norm{S_l}-q=\alpha^N-q,
\end{align}
where we used the fact that $ \norm{\pi'_{l,j}} =\norm{\frac{\mathbb{I}_2+(-1)^{a}Z}{2}
+
\frac{\mathbb{I}_2+(-1)^{a'}X}{2}}=\frac{1}{2}\left(1+\frac{1}{\sqrt{2}}\right)=\alpha$ for all $a,a'\in\{0,1\}$. Hence, $\norm{M_0}<\norm{M_1}$ for all $q\in[0,\alpha^N)$.

In what follows, we demonstrate that for any $k>1$ the operator norm of $M_k$ is upper bounded by the norm of the 1-click operator $M_1$ \eqref{1clicknorm} for a range of the penalty parameter $q$. 

For $k>1$, we need to bound $\norm{M_k}=\norm{ S_1 + S_2 + \cdots S_k} -qk$. For this, we use the result by Popovici and Sebesty\'en \cite{SEBESTYÉN_POPOVICI_2005}. For any set of positive semidefinite matrices (such as the $S_l$), we have
\begin{equation}
\norm{ S_1 + S_2 + \cdots S_k} \le \norm{\Gamma},
\end{equation}
where $\Gamma$ is a $k$-by-$k$ matrix with elements,
\begin{equation}
    \Gamma_{ll'} = \norm{ \sqrt{S_l} \sqrt{S_{l'}} }=\norm{ \sqrt{ \bigotimes_{j=1}^N \pi'_{l,j}} \sqrt{ \bigotimes_{j=1}^N 
 \pi'_{l',j}} }= \norm{ \bigotimes_{j=1}^N \sqrt{\pi'_{l,j}} \sqrt{\pi'_{l',j}}} \leq \prod_{j=1}^{N} \norm{ \sqrt{\pi'_{l,j}} \sqrt{\pi'_{l',j}} }
\end{equation}
where for the second equality we have used the fact the square root of the tensor product is the tensor product of the square roots. 

Observe that whenever $l=l'$ we have, 
\begin{equation}
    \norm{\sqrt{\pi'_{l,j}} \sqrt{\pi'_{l,j}} } = \norm{\pi'_{l,j}} = \frac{1}{2}\left(1+\frac{1}{\sqrt{2}}\right) = \alpha.
\end{equation}
However, when $l\neq l'$, $S_l$ and $S_l'$ differ in the type of $\pi_{l,j^\ast}$ on at least one tensor factor $j^\ast$, and for this particular factor, it is easy to verify that,
\begin{equation}
\norm{\sqrt{\pi'_{l,j^\ast}} \sqrt{\pi'_{l',j^\ast}} } \leq     \frac{1}{8} \left(2+\sqrt{2}+\sqrt{4 \sqrt{2}-2}\right) = \beta \approx 0.665813.
\end{equation}
The rest of the norms are bounded by $\alpha$, and therefore in general we have,
\begin{equation}
    \Gamma_{l,l'} \leq\begin{cases}\alpha^N, \ & \text{if} \ l=l', \\
    \alpha^{N-1}\beta, \  & \text{if} \ l\neq l',
    \end{cases}
\end{equation}
This implies that the matrix $\Gamma$ is upper bounded element-wise by the matrix
\begin{equation}
G_k = \begin{pmatrix}
\alpha^N & \alpha^{N-1}\beta & \alpha^{N-1}\beta & \cdots & \alpha^{N-1}\beta \\
\alpha^{N-1}\beta & \alpha^N & \alpha^{N-1}\beta & \cdots & \alpha^{N-1}\beta \\
\vdots & & \ddots & & \vdots \\
\alpha^{N-1}\beta & \alpha^{N-1}\beta & \alpha^{N-1}\beta & \cdots & \alpha^N
\end{pmatrix} = \left( \alpha^N - \alpha^{N-1}\beta \right) \mathbb{I}_k + \alpha^{N-1}\beta\begin{pmatrix}
1 & 1 & 1 & \cdots & 1 \\
1 & 1 & 1 & \cdots & 1 \\
\vdots & & \ddots & & \vdots \\
1 & 1 & 1 & \cdots & 1
\end{pmatrix},
\end{equation}
where $\mathbb{I}_k$ is the identity operator on $\mathbb{C}^k$.
Since every element of $\Gamma$ and $G_k$ is non-negative, the element-wise inequality implies $\norm{\Gamma} \le \norm{G_k}$ (see Lemma \ref{lem:elementwise}) so that, 
\begin{align} \nonumber
\norm{M_k} \leq \norm{\Gamma} - qk& \leq \norm{G_k} - qk = \left( \alpha^N - \alpha^{N-1}\beta \right) + k\alpha^{N-1}\beta - qk, \\ 
&= \alpha^N + (k-1)\alpha^{N-1}\beta - qk.
\end{align}
which implies that for $q\in[{\alpha^{N-1}\beta}, {\alpha^N})$, we have,
\begin{equation} \label{inequality}
    \norm{M_k}\leq \alpha^N - q = \norm{M_1}.
\end{equation}
as desired. Plugging \eqref{inequality} back into \eqref{CHSHNqgame4JM}, we obtain \eqref{CHSHmaxJM}.

\section{Robustness}\label{app:robustness}

Assume that we have robust self-testing statement for $N$ parallel repetitions of CHSH or BB84, parametrised by some $\varepsilon > 0$, the gap between the maximal and actual Bell inequality violation. In particular, assume that we have that for all purifications $\ket{\psi}_{\tilde{A}\tilde{B}P} \in \cH_{\tilde{A}} \otimes \cH_{\tilde{B}} \otimes \cH_P$ of all physical states $\rho_{\tilde{A}\tilde{B}P} \in \mathcal{B}(\mathcal{H}_{\tilde{A}} \otimes \mathcal{H}_{\tilde{B}})$ and physical measurements $(A^{\bm{x}}_{\bm{a}})_{\bm{a}}$ on $\mathcal{H}_{\tilde{A}}$ and $(B^{\bm{y}}_{\bm{b}})_{\bm{b}}$ on $\mathcal{H}_{\tilde{B}}$ that are $\varepsilon$-compatible with the target distribution, there exist isometries $V_{\tilde{A}} : \mathcal{H}_{\tilde{A}} \to \bC^{2^N} \otimes \cH_{A'}$ and $V_{\tilde{B}}: \cH_{\tilde{B}} \to \bC^{2^N} \otimes \cH_{B'}$ such that
\begin{equation}
\norm{ (V_{\tilde{A}} \otimes V_{\tilde{B}} \otimes \I_P)( \tilde{A}^{\bm{x}}_{\bm{a}} \otimes \tilde{B}^{\bm{y}}_{\bm{b}} \otimes \I_P )\ket{\psi}_{\tilde{A}\tilde{B}P}  - ( {A}^{\bm{x}}_{\bm{a}} \otimes \tilde{B}^{\bm{y}}_{\bm{b}} ) \ket{\phi^+}^{\otimes N} \otimes \ket{ \text{junk}}_{A'B'P} } \le f(N,\varepsilon),
\end{equation}
where $\ket{\phi^+}$ is the maximally entangled two-qubit state, ${A}^{\bm{x}}_{\bm{a}}$ and ${B}^{\bm{y}}_{\bm{b}}$ are the target measurements, and $f(N,\varepsilon)$ is some function of $N$ and $\varepsilon$. Summing up over ${\bm{b}}$, this implies that
\begin{equation}
(V_{\tilde{A}} \otimes V_{\tilde{B}} \otimes \I_P) (\tilde{A}^{\bm{x}}_{\bm{a}} \otimes \I_{\tilde{B}} \otimes \I_P) \ket{\psi}_{\tilde{A}\tilde{B}P} = ({A}^{\bm{x}}_{\bm{a}}  \otimes \I_{B} ) \ket{\phi^+}^{\otimes N} \otimes \ket{ \text{junk}}_{A'B'P} + \ket{\chi^{\bm{x}}_{\bm{a}}},
\end{equation}
where $\ket{\chi^{\bm{x}}_{\bm{a}}} \in \cH_{A} \otimes \cH_{B} \otimes \cH_{A'} \otimes \cH_{B'} \otimes \cH_P=\bC^{2^N} \otimes \bC^{2^N} \otimes \cH_{A'} \otimes \cH_{B'} \otimes \cH_P$ and $\norm{ \ket{\chi^{\bm{x}}_{\bm{a}}} } \le f(N,\varepsilon)$. In terms of density operators (dropping the Hilbert space labels), we have
\begin{equation} \label{imp}
\begin{split}
(V_{\tilde{A}} \otimes V_{\tilde{B}} \otimes \I_P) & \left. (\tilde{A}^{\bm{x}}_{\bm{a}} \otimes \I_{\tilde{B}} \otimes \I_P) \ketbra{\psi}{\psi} (\tilde{A}^{\bm{x}}_{\bm{a}} \otimes \I_{\tilde{B}} \otimes \I_P) (V^\dagger_{\tilde{A}} \otimes V^\dagger_{\tilde{B}} \otimes \I_P) \right. \\
& \left. = ({A}^{\bm{x}}_{\bm{a}}  \otimes \I_{B} ) (\ketbra{\phi^+}{\phi^+})^{\otimes N} ({A}^{\bm{x}}_{\bm{a}}  \otimes \I_B ) \otimes \ketbra{ \text{junk} }{ \text{junk} } + \ketbra{\chi^{\bm{x}}_{\bm{a}}}{\chi^{\bm{x}}_{\bm{a}}} \right. \\
& \left. + \ket{\chi^{\bm{x}}_{\bm{a}}}(\bra{\phi^+})^{\otimes N} ({A}^{\bm{x}}_{\bm{a}}  \otimes \I_{B} ) \otimes \bra{ \text{junk} } + ({A}^{\bm{x}}_{\bm{a}}  \otimes \I_{B} )(\ket{\phi^+})^{\otimes N} \otimes \ket{ \text{junk} } \bra{\chi^{\bm{x}}_{\bm{a}}},
\right.
\end{split}
\end{equation}
%
The norm of the second term is bounded by
\begin{equation}
    \norm{\ketbra{\chi^{\bm{x}}_{\bm{a}}}{\chi^{\bm{x}}_{\bm{a}}} }=\norm{ \ket{\chi^{\bm{x}}_{\bm{a}}} }^2 \le f^2(N,\varepsilon)
\end{equation}
The norm of the third term,
\begin{equation}
    \norm{\ket{\chi^{\bm{x}}_{\bm{a}}}(\bra{\phi^+})^{\otimes N} ({A}^{\bm{x}}_{\bm{a}}  \otimes \I_{B} ) \otimes \bra{ \text{junk} }} \leq \norm{ \ket{\chi^{\bm{x}}_{\bm{a}}} }\norm{ (\bra{\phi^+})^{\otimes N} ({A}^{\bm{x}}_{\bm{a}}  \otimes \I_{B} ) \otimes \bra{ \text{junk} } }\leq \norm{\ket{\chi^{\bm{x}}_{\bm{a}}}} \le f(N,\varepsilon), 
\end{equation}
and similarly for the last term,
\begin{equation}
    \norm{(\tilde{A}^{\bm{x}}_{\bm{a}}  \otimes \I_{B} )(\ket{\phi^+})^{\otimes N} \otimes \ket{ \text{junk} } \bra{\chi^{\bm{x}}_{\bm{a}}}}\leq \norm{\ket{\chi^{\bm{x}}_{\bm{a}}}} \le f(N,\varepsilon).
\end{equation}
Therefore, the norm of the sum of the last three terms is bounded by $2f(N,\varepsilon) + f^2(N,\varepsilon)$

Tracing out Alice's $N$ qubits as well as $A'$ and $P$ in Eq.~\eqref{imp}, we get
\begin{equation}\label{eq:Vrho_axV}
V_{\tilde{B}} \tilde{\rho}^{(\tilde{B})}_{\bm{a}|\bm{x}} V^\dagger_{\tilde{B}} = \rho^{(B)}_{\bm{a}|\bm{x}} \otimes \sigma + \eta_{\bm{a}|\bm{x}},
\end{equation}
where
\begin{equation}
\begin{split}
    \tilde{\rho}^{(\tilde{B})}_{\bm{a}|\bm{x}} & \left. =\tr_{A,A',P}[(\tilde{A}^{\bm{x}}_{\bm{a}} \otimes \I_{\tilde{B}} \otimes \I_P) \ketbra{\psi}{\psi} (\tilde{A}^{\bm{x}}_{\bm{a}} \otimes \I_{\tilde{B}} \otimes \I_P)], \right. \\
    \tilde{\rho}_{\bm{a}|\bm{x}} & \left. =\tr_{A,A',P}[({A}^{\bm{x}}_{\bm{a}}  \otimes \I_{B} ) (\ketbra{\phi^+}{\phi^+})^{\otimes N} ({A}^{\bm{x}}_{\bm{a}}  \otimes \I_{B} )], \right. \\
    \sigma & \left. = \tr_{A'P}( \ketbra{ \text{junk} }{\text{junk} } ), \right. \\
\eta_{\bm{a}|\bm{x}} & \left. = \tr_{A,A',P}\big[ \ketbra{\chi^{\bm{x}}_{\bm{a}}}{\chi^{\bm{x}}_{\bm{a}}} \big] + \tr_{A,A',P}\big[ \ket{\chi^{\bm{x}}_{\bm{a}}}(\bra{\phi^+})^{\otimes N} ({A}^{\bm{x}}_{\bm{a}}  \otimes \I_{B} ) \otimes \bra{ \text{junk} } \big] \right. \\
& \left. + \tr_{A,A',P}\big[ ({A}^{\bm{x}}_{\bm{a}}  \otimes \I_{B} )(\ket{\phi^+})^{\otimes N} \otimes \ket{ \text{junk} } \bra{\chi^{\bm{x}}_{\bm{a}}} \big].
\right.
\end{split}
\end{equation}
In the following, we bound the norm of $\eta_{\bm{a}|\bm{x}}$. Notice that $\eta_{\bm{a}|\bm{x}}$ is the sum of three terms, all three of the form $\tr_X( \ketbra{\varphi}{\xi} )$, and we already have bounds on the norms of these $\ketbra{\varphi}{\xi}$. To bound the norms of the partial traces, we employ the following lemma.
\begin{lem}
Let $\ket{\varphi}, \ket{\xi} \in \cH_X \otimes \cH_Y$. Then we have $\norm{ \tr_X( \ketbra{\varphi}{\xi} ) } \le \norm{ \ketbra{\varphi}{\xi}} $.
\end{lem}
\begin{proof}
Consider the Schmidt decompositions $\ket{\varphi} = \sum_j \lambda_j \ket{e_j} \otimes \ket{f_j}$ and $\ket{\xi} = \sum_k \eta_k \ket{g_k} \otimes \ket{h_k}$, where $\lambda_j, \eta_k \ge 0$, $\{ \ket{e_j} \}_j$ and $\{ \ket{g_k} \}_k$ are orthonormal sets on $\cH_X$, and $\{ \ket{f_j} \}_j$ and $\{ \ket{h_k} \}_k$ are orthonormal \emph{bases} on $\cH_Y$ (we allow for $\lambda_j =0$ and $\eta_k = 0$ so that $\{ \ket{f_j} \}_j$ and $\{ \ket{h_k} \}_k$ can be complete bases by adequately padding with zero Schmidt coefficients, embedding $\ket{\varphi}$ and $\ket{\xi}$ in a larger Hilbert space if necessary). Then we have
\begin{equation}
\norm{ \ketbra{\varphi}{\xi} } = \sqrt{ \norm{\ketbra{\varphi}{\xi}\ketbra{\xi}{\varphi} } } = \sqrt{ \braket{\xi}{\xi} \norm{ \ketbra{ \varphi }{\varphi} } } = \sqrt{ \braket{\xi}{\xi} \braket{ \varphi}{\varphi}} = \sqrt{ \sum_j \lambda_j^2 \sum_k \eta_k^2 }
\end{equation}
Furthermore,
\begin{equation}
\tr_X( \ketbra{\varphi}{\xi} ) = \tr_X\Big( \sum_{j,k} \lambda_j \eta_k \ketbra{e_j}{g_k} \otimes \ketbra{f_j}{h_k} \Big) = \sum_{j,k} \lambda_j \eta_k \braket{g_k}{e_j} \ketbra{f_j}{h_k}.
\end{equation}
Denoting $T:=\tr_X( \ketbra{\varphi}{\xi})$, we have
\begin{equation}\label{eq:traceTT}
\norm{ \tr_X( \ketbra{\varphi}{\xi} )} = \norm{T} = \sqrt{\norm{T^\dagger T}} \le \sqrt{ \tr( T^\dagger T ) },
\end{equation}
since $T^\dagger T \ge 0$. Explicitly,
\begin{equation}
T^\dagger T = \sum_{j,k,l,m} \lambda_j \eta_k \lambda_l \eta_m \braket{e_j}{g_k} \ketbra{h_k}{f_j} \braket{g_m}{e_l} \ketbra{f_l}{h_m} = \sum_{j,k,m} \lambda^2_j \eta_k \eta_m \braket{e_j}{g_k} \braket{g_m}{e_j} \ketbra{h_k}{h_m}
\end{equation}
and so
\begin{equation}
\tr( T^\dagger T ) = \sum_{j,k} \lambda_j^2 \eta_k^2 \braket{e_j}{g_k} \braket{g_k}{e_j} \le \sum_{j,k} \lambda_j^2 \eta_k^2,
\end{equation}
since the $\ket{e_j}$ and $\ket{g_k}$ vectors are normalised. From Eq.~\eqref{eq:traceTT} we therefore have
\begin{equation}
\norm{ \tr_X( \ketbra{\varphi}{\xi} )} \le \sqrt{ \sum_{j,k} \lambda_j^2 \eta_k^2 } = \norm{ \ketbra{\varphi}{\xi} }
\end{equation}
as desired.
\end{proof}
Putting everything together, we then have

\begin{equation}\label{eq:eta_opnorm_bound}
    \norm{\eta_{\bm{a}|\bm{x}}} \le 2f(N,\varepsilon) + f^2(N,\varepsilon).
\end{equation}
Similarly to the ideal ($\varepsilon = 0$) case, we simplify the notation in Eq.~\eqref{eq:Vrho_axV} and write
\begin{equation}\label{eq:USU}
V \tilde{S}_j V^\dagger = S_j \otimes \sigma + \eta_j,
\end{equation}
where we chose the multi-index $j$ for $a|x$, re-named $\rho_{\bm{a}|\bm{x}} \to S_j$ (similarly for the tilde version) and we also dropped the $B$ index from the isometry. Note that the $\tilde{S}_j$ are not normalised (projective) anymore, which is a slight difference from the notation in the ideal case! In particular, $S_j$ are rank-1 projections multiplied by $\frac{1}{2^{2N}}$.

From Eq.~\eqref{eq:USU}, we have
\begin{equation}
\tilde{S}_j = V^\dagger (S_j \otimes \sigma + \eta_j) V,
\end{equation}
%
%
%
and we want to bound the norm of $\sum_{j=1}^k \tilde{S}_j = \sum_{j=1}^k V^\dagger(S_j \otimes \sigma + \eta_j)V = V^\dagger \Big[ \sum_{j=1}^k (S_j \otimes \sigma + \eta_j) \Big] V$, where $V$ is an isometry. For any operator $M$ and any isometry $V$, we have by the submultiplicativity of the operator norm that $\norm{ V^\dagger M V} \le \norm{V^\dagger} \cdot \norm{M} \cdot \norm{V} = \norm{M}$, since $\norm{V} = \norm{V^\dagger} = 1$ for any isometry. Therefore,
\begin{equation}
\norm{ \sum_{j=1}^k \tilde{S}_j } \le \norm{ \sum_{j=1}^k ({S}_j \otimes \sigma + \eta_j )}.
\end{equation}
Since ${S}_j \otimes \sigma + \eta_j \ge 0$ by Eq.~\eqref{eq:USU}, similarly to the ideal case we can apply the Popovici--Sebesty\'en bound and eventually we need to bound the quantities
\begin{equation}
\begin{split}
& \left. \norm{ \sqrt{ {S}_i \otimes \sigma + \eta_i }\sqrt{ {S}_j \otimes \sigma + \eta_j } } \le \norm{\sqrt{ {S}_i \otimes \sigma + \eta_i }\sqrt{ {S}_j \otimes \sigma + \eta_j } }_F \right. \\
& \left. = \sqrt{ \tr[ ( {S}_i \otimes \sigma + \eta_i ) ({S}_j \otimes \sigma + \eta_j ) ] } = \sqrt{ \tr[ ({S}_i \otimes \sigma) ({S}_j \otimes \sigma) ] + \tr[ ( {S}_i \otimes \sigma ) \eta_j ] + \tr[ \eta_i ( {S}_j \otimes \sigma ) ] + \tr(\eta_i \eta_j) } \right. \\
& \left. \le \sqrt{ \tr( {S}_i {S}_j )  + \tr[ ( {S}_i \otimes \sigma ) \eta_j ] + \tr[ \eta_i ( \tilde{S}_j \otimes \sigma ) ] + \tr(\eta_i \eta_j) }
\right.
\end{split}
\end{equation}
where $\norm{.}_F$ is the Frobenius norm and we used $\tr[ ({S}_i \otimes \sigma) ({S}_j \otimes \sigma) ] = \tr( {S}_i {S}_j ) \tr( \sigma^2) \le \tr( {S}_i {S}_j )$.

By the same arguments as in the ideal case, we have $\tr( {S}_i {S}_i ) = \frac{1}{2^{4N}}$ and $\tr( {S}_i {S}_j ) \le \frac{1}{2^{4N+1}}$ for all $i \neq j$.

The second term can be bounded by
\begin{equation}
\tr[ ( {S}_i \otimes \sigma ) \eta_j ] \le \abs{ \tr[ ( {S}_i \otimes \sigma ) \eta_j ] } \le \norm{{S}_i \otimes \sigma }_1 \norm{ \eta_j} \le \frac{1}{2^{2N}} \norm{ \eta_j },
\end{equation}
where $\norm{.}_1$ is the trace norm and we used the H\"older inequality, and the same bound holds for the third term.

We bound the last term similarly,
\begin{equation}\label{eq:tr_etaeta}
\tr( \eta_i \eta_j ) \le \abs{ \tr( \eta_i \eta_j ) } \le \norm{\eta_i}_1 \norm{\eta_j}.
\end{equation}
Note that $\eta_i$ is of the form
\begin{equation}
\eta_i = \tr_A \left( \ketbra{\chi_i}{\chi_i} + \ketbra{\chi_i}{\psi_i} + \ketbra{\psi_i}{\chi_i} \right),
\end{equation}
where $\ket{\psi_i} = ({A}^{\bm{x}}_{\bm{a}}  \otimes \I_{B} )\ket{\phi^+}^{\otimes N} \otimes \ket{ \text{junk}}$ (remembering that we are using $i$ as a multi-index). Using the fact that the trace norm is decreasing under the partial trace \cite{Rastegin2012} and the triangle inequality, we get
\begin{equation}
\norm{ \eta_i }_1 \le \norm{ \ketbra{\chi}{\chi} }_1 + \norm{ \ketbra{\chi}{\psi_i} }_1 + \norm{ \ketbra{\psi_i}{\chi} }_1.
\end{equation}
Note that the trace norm of an operator $\ketbra{a}{b}$ is given by
\begin{equation}
\norm{ \ketbra{a}{b} }_1 = \tr( \sqrt{ \ket{b}\!\braket{a}{a} \!\bra{b} } ) = \sqrt{\braket{a}{a}} \tr(\frac{ \ketbra{b}{b} }{ \sqrt{ \braket{b}{b} } }) = \sqrt{\braket{a}{a}} \sqrt{\braket{b}{b}} = \norm{ \ket{a} } \cdot \norm{\ket{b}}
\end{equation}
Therefore,
\begin{equation}
\norm{ \eta_i }_1 \le \norm{ \ket{ \chi_i } }^2 + 2 \norm{\ket{\chi_i}} \norm{\ket{\psi_i} } \le f^2(N,\varepsilon) + 2f(N,\varepsilon)
\end{equation}

We had the same bound for the operator norm in Eq.~\eqref{eq:eta_opnorm_bound}, and so  from Eq.~\eqref{eq:tr_etaeta} we get
\begin{equation}
\tr( \eta_i \eta_j ) \le \left[ f^2(N,\varepsilon) + 2f(N,\varepsilon) \right]^2.
\end{equation}

Overall, we obtain
\begin{equation}
\norm{ \sqrt{ {S}_i \otimes \sigma + \eta_i }\sqrt{ {S}_j \otimes \sigma + \eta_j } } \le \frac{1}{2^{2N}} \sqrt{ \frac12 + \delta_{ij} }
\end{equation}
for $i\neq j$, and%
\begin{equation}
\norm{ \sqrt{ {S}_i \otimes \sigma + \eta_i }\sqrt{ {S}_i \otimes \sigma + \eta_i } } \le \frac{1}{2^{2N}} \sqrt{ 1 + \delta_{ii} },
\end{equation}
where
\begin{equation}
\begin{split}\label{eq:delta}
0 \le \delta_{ij} & \left. \le 2^{2N+1} \left[ f^2(N,\varepsilon) + 2f(N,\varepsilon) \right] + 2^{4N} \left[ f^2(N,\varepsilon) + 2f(N,\varepsilon) \right]^2
\right.
\end{split}
\end{equation}
for all $i,j$. Let us denote the largest of these by $\delta \equiv \max_{i,j}\{ \delta_{ij} \}$. Then
\begin{equation}
\norm{ \sqrt{ {S}_i \otimes \sigma + \eta_i }\sqrt{ {S}_j \otimes \sigma + \eta_j } } \le \frac{1}{2^{2N}} \sqrt{ \frac12 + \delta } \le \frac{1}{2^{2N}} \left( \frac{1}{\sqrt{2}} + \sqrt{\delta} \right)
\end{equation}
for $i\neq j$, and%
\begin{equation}
\norm{ \sqrt{ {S}_i \otimes \sigma + \eta_i }\sqrt{ {S}_i \otimes \sigma + \eta_i } } \le \frac{1}{2^{2N}} \sqrt{ 1 + \delta } \le \frac{1}{2^{2N}} \left( 1 + \sqrt{\delta} \right),
\end{equation}
where we used that for $a,b \ge 0$ we have $\sqrt{a+b} \le \sqrt{ a + b + 2\sqrt{a}\sqrt{b} } = \sqrt{ ( \sqrt{a} + \sqrt{b} )^2 } = \sqrt{a} + \sqrt{b}$.

By the same arguments as for the ideal case, the maximal eigenvalue of all possible $C_{\vec{\mathbf{b}}}$ operators is dictated by the norm of the $k$-by-$k$ matrix
\begin{equation}
\begin{split}
F_k & \left. =
\begin{pmatrix} 
    1 + \sqrt{\delta} & \frac{1}{\sqrt{2}} + \sqrt{\delta} & \frac{1}{\sqrt{2}} + \sqrt{\delta} & \cdots & \frac{1}{\sqrt{2}} + \sqrt{\delta} \\
    \frac{1}{\sqrt{2}} + \sqrt{\delta} & 1 + \sqrt{\delta} & \frac{1}{\sqrt{2}} + \sqrt{\delta} & \cdots & \frac{1}{\sqrt{2}} + \sqrt{\delta} \\
    \vdots & & \ddots & & \vdots \\
    \frac{1}{\sqrt{2}} + \sqrt{\delta} & \frac{1}{\sqrt{2}} + \sqrt{\delta} & \frac{1}{\sqrt{2}} + \sqrt{\delta} & \cdots & 1 + \sqrt{\delta}
\end{pmatrix}
=
\begin{pmatrix}
    1 & \frac{1}{\sqrt{2}} & \frac{1}{\sqrt{2}} & \cdots & \frac{1}{\sqrt{2}} \\
    \frac{1}{\sqrt{2}} & 1 & \frac{1}{\sqrt{2}} & \cdots & \frac{1}{\sqrt{2}} \\
    \vdots & & \ddots & & \vdots \\
    \frac{1}{\sqrt{2}} & \frac{1}{\sqrt{2}} & \frac{1}{\sqrt{2}} & \cdots & 1
\end{pmatrix}
+
\begin{pmatrix}
    \sqrt{\delta} & \sqrt{\delta} & \sqrt{\delta} & \cdots & \sqrt{\delta} \\
    \sqrt{\delta} & \sqrt{\delta} & \sqrt{\delta} & \cdots & \sqrt{\delta} \\
    \vdots & & \ddots & & \vdots \\
    \sqrt{\delta} & \sqrt{\delta} & \sqrt{\delta} & \cdots & \sqrt{\delta} \\
\end{pmatrix}
\right. \\
& \left. = 
\left( 1 - \frac{1}{\sqrt{2}} \right) \I_k +
\left( \frac{1}{\sqrt{2}} + \sqrt{\delta} \right)
\begin{pmatrix}
    1 & 1 & 1 & \cdots & 1 \\
    1 & 1 & 1 & \cdots & 1 \\
    \vdots & & \ddots & & \vdots \\
    1 & 1 & 1 & \cdots & 1
\end{pmatrix}
\right.
\end{split}
\end{equation}
maximised over all possible values of $k$. The norm of $F_k$ is
\begin{equation}
\norm{F_k} = \left( 1 - \frac{1}{\sqrt{2}} \right) + k \left( \frac{1}{\sqrt{2}} + \sqrt{\delta} \right)
\end{equation}
Therefore, as in the ideal case, the maximal eigenvalue of all the $C_{\vec{\mathbf{b}}}$ is bounded by the maximum over $k \in \{1, 2, \ldots, 2^N\}$ of
\begin{equation}
f_k \equiv \left( 1 - \frac{1}{\sqrt{2}} \right) + k \left( \frac{1}{\sqrt{2}} + \sqrt{\delta} \right) + -kq.
\end{equation}
For simplicity, we choose $q = \left( \frac{1}{\sqrt{2}} + \sqrt{\delta} \right)$, and therefore
\begin{equation}
f_k = 1- \frac{1}{\sqrt{2}}
\end{equation}
for all $k$, leading to the JM bound
\begin{equation}
    \mathcal{B}^{ N}_{AB_1}(q)\underset{\text{JM}({B_1})} \le \frac{ 1- \frac{1}{\sqrt{2}} }{2^N},
\end{equation}
noting that the inequality to check depends on $q$, and therefore on $N$ and $\varepsilon$ through $\delta$.

To obtain the critical detection efficiency for violating this inequality, we take two approaches. First, an idealized (but somewhat unrealistic) approach that $A$ and $B_1$ reaches the $q$-penalized $N$-product BB84 value of $\mathcal{B}^{ N}_{AB_1}(q)=(1-q)\eta$ with detection efficiency $\eta$. This approach assumes that the only source of imperfection is the detection efficiency, which is not to be expected when the $N$-product CHSH inequality is only violated up to $\varepsilon$. This approach leads to the critical detection efficiency
\begin{equation}
    \eta = \frac{1}{2^N} \frac{ 1 - \frac{1}{\sqrt{2}} }{ 1 - q } = \frac{1}{2^N} \frac{ 1 - \frac{1}{\sqrt{2}} }{ 1 - \frac{1}{\sqrt{2}} - \sqrt{ \delta } } > \frac{1}{2^N},
\end{equation}
since $\delta > 0$. At the same time,
\begin{equation}
    \eta \leq \frac{1}{2^N} \frac{ 1 - \frac{1}{\sqrt{2}} }{ 1 - \frac{1}{\sqrt{2}} - \left[ 2^{2N+1} \left[ f^2(N,\varepsilon) + 2f(N,\varepsilon) \right] + 2^{4N} \left[ f^2(N,\varepsilon) + 2f(N,\varepsilon) \right]^2 \right]^{\frac12} }
\end{equation}
by Eq.~\eqref{eq:delta}.

In the other approach, we assume that $A$ and $B_1$ reaches the $q$-penalized $N$-product BB84 value of $\mathcal{B}^{ N}_{AB_1}(q)=(1-\frac{\varepsilon}{\alpha^N} - q)\eta$ with detection efficiency $\eta$. This is plausible if e.g.~$\varepsilon$ is linearly related to a state visibility parameter. This approach leads to the critical detection efficiency
\begin{equation}
\begin{split}
    \eta & \left. = \frac{1}{2^N} \frac{ 1 - \frac{1}{\sqrt{2}} }{ 1 - q } = \frac{1}{2^N} \frac{ 1 - \frac{1}{\sqrt{2}} }{ 1 - \frac{1}{\sqrt{2}} - \frac{\varepsilon}{\alpha^N} - \sqrt{ \delta } } \right. \\
    & \left. \leq \frac{1}{2^N} \frac{ 1 - \frac{1}{\sqrt{2}} }{ 1 - \frac{1}{\sqrt{2}} - \frac{\varepsilon}{\alpha^N} - \left[ 2^{2N+1} \left[ f^2(N,\varepsilon) + 2f(N,\varepsilon) \right] + 2^{4N} \left[ f^2(N,\varepsilon) + 2f(N,\varepsilon) \right]^2 \right]^{\frac12} }
\right.
\end{split}
\end{equation}
similarly to above.

\end{document}